\documentclass[journal]{IEEEtran}

\usepackage{cite}
\usepackage{amsmath,amssymb,amsfonts}
\usepackage{graphicx}
\usepackage{textcomp}
\usepackage{xcolor}
\usepackage{mathrsfs}
\usepackage[noend]{algpseudocode}
\usepackage{algorithmicx,algorithm}
\usepackage{epstopdf}
\usepackage{pifont}
\usepackage{amsthm}

\usepackage{bm}
\usepackage{subfigure}
\usepackage{setspace}
\usepackage{verbatim}
\usepackage[bookmarks,colorlinks]{hyperref}
\usepackage{acronym}
\usepackage[marginal]{footmisc}
\def\BibTeX{{\rm B\kern-.05em{\sc i\kern-.025em b}\kern-.08em
    T\kern-.1667em\lower.7ex\hbox{E}\kern-.125emX}}

\newtheorem{lemma}{Lemma}

\newtheorem{proposition}{Proposition}

\newtheoremstyle{noparens}
  {}{}
  {\itshape}{}
  {\bfseries}{.}
  { }
  {\thmname{#1}\thmnumber{ #2}\mdseries\thmnote{ #3}}

\theoremstyle{noparens}

\acrodef{crb}[CRB]{Cram$ {\rm \acute{e}} $r-Rao lower bound}
\acrodef{snr}[SNR]{signal-to-noise rate}
\acrodef{srl}[SRL]{statistical resolution limit}
\acrodef{music}[MUSIC]{Multiple Signal Classification}
\acrodef{fim}[FIM]{Fisher information matrix}
\acrodef{rmse}[RMSE]{root mean square error}

\begin{document}

\title{Fundamental Limit of Angular Resolution in Partly Calibrated Arrays with Position Errors}

\author{Guangbin Zhang, Yan Wang, Tianyao Huang and Yonina C. Eldar
\thanks{\quad 
This work was supported by the National Natural Science Foundation of China under Grants 62271053, 62331007, 62171259, 62401063.
G. Zhang is with Yangtze Delta Region Academy of Beijing Institute of Technology, Jiaxing, Zhejiang, China (e-mail:zgb02@qq.com).
Y. Wang is with Radar Research Laboratory, School of Information and Electronics, Beijing Institute of Technology, Beijing, China (e-mail:yan\_wang@bit.edu.cn).
T. Huang is with the School of Computer and Communication Engineering, University of Science and Technology Beijing, Beijing, China (e-mail:huangtianyao@ustb.edu.cn).
Yonina C. Eldar is with the Department of Computer Science and Applied Mathematics, Weizmann Institute of Science, Israel (e-mail:yonina.eldar@weizmann.ac.il). 
Y. Wang is the corresponding author.}
} 


\maketitle

\begin{abstract}

We consider high angular resolution detection using distributed mobile platforms implemented with so-called \emph{partly calibrated arrays}, where position errors between subarrays exist and the counterparts within each subarray are ideally calibrated. 
Since position errors between antenna arrays affect the coherent processing of measurements from these arrays, it is commonly believed that its angular resolution is influenced. 
A key question is whether and how much the angular resolution of partly calibrated arrays is affected by the position errors, in comparison with ideally calibrated arrays.  
To address this fundamental problem, we theoretically illustrate that partly calibrated arrays approximately achieve high angular resolution. 
Our analysis uses a special characteristic of Cram$ {\rm \acute{e}} $r-Rao lower bound (CRB) w.r.t. the source separation: When the source separation increases, the CRB first declines rapidly, then plateaus out, and the turning point is close to the angular resolution limit.
This means that the turning point of CRB can be used to indicate angular resolution.
We then theoretically analyze the declining and plateau phases of CRB, and explain that the turning point of CRB in partly calibrated arrays is close to the angular resolution limit of distributed arrays without errors, demonstrating high resolution ability.
This work thus provides a theoretical guarantee for the high-resolution performance of distributed antenna arrays in mobile platforms.

\end{abstract}

\begin{IEEEkeywords}
Partly calibrated arrays, angular resolution, the declining and plateau phases of CRB, the turning point of CRB.
\end{IEEEkeywords}

\IEEEpeerreviewmaketitle

\section{Introduction}

High angular resolution is desired to achieve precise target detection in applications such as radar, sonar, and astronomy.
Antenna array is a common tool for direction finding.
Since the angular resolution of an antenna array is inversely proportional to the aperture of the array \cite{van2004optimum}, large aperture arrays are used to achieve high resolution.
However, for a single array, large aperture means high system complexity, high cost and poor mobility, which restricts the scope of application.

A promising solution is to instead use multiple distributed arrays of small apertures and fuse their measurements coherently, known as `distributed arrays' \cite{Heimiller83}. 
Ideally, distributed arrays achieve high angular resolution inversely proportional to the whole array aperture with lower system complexity.
However, for distributed arrays loaded on mobile platforms such as unmanned aerial vehicles (UAVs), which are common in the low-altitude economy, drone swarms, and similar applications, it is difficult to locate the arrays accurately in real time due to the mobility of the platforms, and position errors between subarrays are unavoidable.
Such distributed arrays with unknown inter-subarray position errors (and with no or calibrated intra-subarray position errors) are called {partly calibrated arrays \cite{PCA96,RARE02,exRARE04,Lei10,Parvazi11,Steffens14,COBRAS18,Non-Coherent18,zhang2023decentralized,zhang2023fundamental,su2023direction,zgb23}.
In contrast, distributed arrays with exactly known array positions are called fully calibrated arrays \cite{stoica1989music}.
A main concern is whether partly calibrated arrays could achieve the same (or similar) angular resolution as fully calibrated arrays under the negative influence of unknown errors.
Note that partly calibrated arrays in this paper refer to distributed arrays with position errors between subarrays, and do not include other types of errors such as gain and phase errors \cite{6621815,liao2011direction} or clock synchronization errors \cite{9994246}.

Though the performance of direction finding with partly calibrated arrays is intuitively inferior to the counterpart with error-free arrays, some existing algorithms experimentally show that under certain assumptions the former achieve similar angular resolution as the latter.
These assumptions include, for example, high \ac{snr}, enough snapshots and uncorrelated source signals \cite{RARE02,exRARE04}, and isotropic linear arrays with the same topology \cite{zgb23}. 
However, the high-resolution ability depends on specific algorithms and their assumptions, and the scope of application is limited.
Whether errors seriously degrade the resolution or not is still an open problem in more general scenarios, and is hard to solve only from the perspective of algorithm design. 



Inspired by the positive empirical results, we aim to theoretically analyze the fundamental limit of angular resolution in partly calibrated arrays. 
The first step is to quantify angular resolution for general arrays. 
However, the typical Rayleigh criterion, where the central maximum of one source’s beamforming result coincides with the first minimum of the other, struggles to effectively explain the resolution of distributed arrays with errors. 
This is because the Rayleigh criterion is under the assumption of error-free scenarios. 
When errors are present, the beamforming results are significantly affected by the unknown errors and exhibit irregular behavior, making it difficult to characterize the resolution performance as in the ideal cases.
This motivates us to explore alternative resolution metrics.

To this end, the \ac{srl} was proposed in this context and studied extensively \cite{cox1973resolving,sharman1985resolving,el2011statistical,2014Ren,sun2017statistical,thameri2018statistical,Lee92,Smith_resolution05,abeida2022refinement}.
\ac{srl} is empirically defined as the minimum separation between the parameter of interest that makes two closely spaced signals distinguishable.  
Several criteria are introduced to describe \ac{srl}, which are mainly divided into spectrum based \cite{cox1973resolving,sharman1985resolving}, detection based \cite{el2011statistical,2014Ren,sun2017statistical,thameri2018statistical} and estimation based \cite{Lee92,Smith_resolution05,abeida2022refinement,wan2023} criteria.
However, the spectrum based resolution criterion is not perfectly suitable for partly calibrated arrays with position errors, since the distortion of spectral peaks caused by errors can significantly complicates the analysis.
Otherwise, spectrum based and detection based criteria depend on specific estimation algorithms and hypothesis testing strategies, respectively.
Estimation based criteria use estimation accuracy limit, \ac{crb}, to characterize the resolution limit, which is independent of specific algorithms or detection strategies.
However, typical \ac{crb} based criteria rely on high \ac{snr} and no modeling or signal mismatch \cite{Smith_resolution05}, which are not directly applicable in distributed arrays with position errors.
Recently, a resolution criterion based on Gaussian process was proposed \cite{usmani2024longitudinal}. However, it primarily targets optical three-dimensional imaging, making it difficult to apply directly to radar signal processing.

The main factor affecting the resolution is array aperture.
To better quantify the angular resolution in a way less dependent on \ac{snr} and the error-free assumptions, we exploit a characteristic of \ac{crb} with respect to angular separation. 
Particularly, the \ac{crb} curve first shows a rapid decline along with the increase of the source separation $ \Delta\omega $ from zero to a turning point. After that turning point, the curve displays minor fluctuations, and soon converges to some fixed level. 
The turning point is close to the angular resolution limit $ \Omega $ \cite{Smith_resolution05}. 
We explain the reason of this phenomenon as follows: 1) when sources are closely placed and are unresolvable ($ \Delta\omega \leqslant\Omega $), the estimation accuracy is poor and the \ac{crb} is high; in this region, as the separation increases, the \ac{crb} declines, indicating the significant improvement on the estimation accuracy; when the sources becomes resolvable ($ \Delta\omega\geqslant\Omega $), the \ac{crb} tends to be stable.
This shows the rationality of using the \ac{crb} curve's turning point as an indicator of angular resolution. 

We then show that in partly calibrated arrays the \ac{crb} curve's turning point is inversely proportional to the whole aperture of the array by analyzing the behavior of \ac{crb} in two regions:
in the unresolvable part ($ \Delta\omega\ll\Omega $), we show that the \ac{crb} declines polynomially 
with respect to $ \Delta\omega$ and theoretically calculate the decline speed;
in the resolvable part ($ \Delta\omega\geqslant\Omega $), we illustrate that the partial derivative of \ac{crb} with respect to $ \Delta\omega $ is close to zero. 
These two behaviours in the declining and plateau phases confirm the existence and location of the turning point, and consequently indicate angular resolution. 


Our main contributions are summarized as follows:
\begin{itemize}
\item We propose a new criterion that uses the turning point of \ac{crb} to indicate angular resolution;
\item We use the proposed criterion to demonstrate that partly calibrated arrays achieve high angular resolution similar as fully calibrated arrays, both inversely proportional to the whole array aperture.
\end{itemize}
The above conclusion provides an important theoretical guarantee for the high-resolution performance of distributed antenna arrays in mobile platforms.
To our best knowledge, this work is the first to theoretically explain the high resolution performance limit of distributed arrays in the presence of position errors between subarrays.

The rest of this paper is organized as follows.
Section~\ref{sec:CRB_resolution} reviews the related works.
Section~\ref{sec:signal} provides the signal models of fully and partly calibrated arrays. 
Section~\ref{sec:indicate} introduces our main contributions of indicating the angular resolution of partly calibrated arrays using the proposed criterion. 
In Section~\ref{sec:analysis}, we detail the proofs of our main contributions.
Numerical simulations are given in Section~\ref{sec:simulation} to verify the analysis, followed by a conclusion in Section~\ref{sec:conclusion}.

We use $ \mathbb{Z} $, $ \mathbb{R} $ and $ \mathbb{C} $ to denote the sets of integer, real and complex numbers, respectively.
The expectation of a random variable $ \cdot $ is written as $ \mathbb{E}[\cdot] $.
Uppercase boldface letters denote matrices (e.g. $ \bm{A} $) and lowercase boldface letters denote vectors (e.g. $ \bm{a} $). The $(m,n)$-th element of a matrix $ \bm{A} $ is denoted by $ [\bm{A}]_{m,n} $, and the $n$-th column is represented by $ [\bm{A}]_{n} $.  
We use $ {\rm trace}(\cdot) $ to indicate the trace of a matrix and $ {\rm diag}(\bm{a}) $ to represent a matrix with diagonal elements given by $ \bm{a} $.
The  conjugate, transpose, and  conjugate transpose operators are denoted by $ ^*,^T,^H $, respectively. 
The amplitude of a scalar and the $ l_2 $ norm of a vector are represented by $ |\cdot| $ and $ \Vert \cdot \Vert_2 $, respectively.
The cardinality of a set $ \mathcal{N} $ is represented by $ |\mathcal{N}| $.
The Hadamard product is written as $ \odot $, the semi-definite operator is denoted by $ \succcurlyeq $, and the definition symbol is defined as $ \equiv $. 
We denote the imaginary unit for complex numbers by $ j=\sqrt{-1} $.


\section{CRB as a resolution metric}
\label{sec:CRB_resolution}

In this section, we review the related works of \ac{crb} based resolution criteria and analyze their application in indicating angular resolution of partly calibrated arrays.

To clarify the resolvability of closely spaced signals in a given scenario, \ac{srl} is an efficient typical tool that received wide attention \cite{cox1973resolving,sharman1985resolving,el2011statistical,2014Ren,sun2017statistical,thameri2018statistical,Lee92,Smith_resolution05,abeida2022refinement,yau1992worst,clark1995resolvability,liu2007statistical}. 
Particularly, \ac{srl} is defined as \emph{the minimum distance between two closely spaced signals embedded in an additive noise that allows a correct resolvability/parameter estimation} \cite{el2011statistical}.
One of the main techniques to describe and derive \ac{srl} is based on estimation accuracy since the resolved signals intuitively have higher estimation accuracy than the corresponding unresolved ones.
\ac{crb} is widely used to characterize the upper bound of estimation accuracy, and therefore it is natural to combine \ac{srl} and \ac{crb} to explain the resolution limit. 

Particularly, assume that there are two signals with frequencies $ \bm{f}=[f_1,f_2]^T $ mixed together and denote the resolution limit of $ \bm{f} $ by $ \Omega_f $. 
Let the \ac{crb} of $ \bm{f} $ be written as
\begin{align}
\label{equ:CRB_f}
{\rm CRB}(\bm{f})=\begin{bmatrix}
{\rm CRB}(f_1) & {\rm CRB}(f_1,f_2) \\
{\rm CRB}(f_1,f_2) & {\rm CRB}(f_2)
\end{bmatrix}.
\end{align}
The average \ac{crb} of frequency $ \bm{f} $ is defined as 
\begin{align}
\label{equ:average_CRB}
\overline{\rm CRB}(\bm{f})=\frac{1}{2}\big({\rm CRB}(f_1)+{\rm CRB}(f_2)\big).
\end{align}
Unless otherwise emphasized, the \ac{crb}s mentioned after this section refer to the average \ac{crb} (omitting the horizontal bar).

Existing related works construct the correlation between resolution limit and \ac{crb} by some criteria.
In \cite{Lee92}, Lee criterion states that: \emph{two signals are resolvable w.r.t. the frequencies if the maximum standard deviation of each frequency estimate is less than half the difference between $ f_1 $ and $ f_2 $}, shown as $ \Omega_f=2{\rm max}\left\{\sqrt{{\rm CRB}(f_1)},\sqrt{{\rm CRB}(f_2)}\right\} $. 
This criterion ignores the coupling between the parameters $ f_1,f_2 $, i.e., $ {\rm CRB}(f_1,f_2) $ in \eqref{equ:CRB_f}. 
To this end, Smith criterion in \cite{Smith_resolution05} states that: \emph{two signals are resolvable w.r.t. the frequencies if the difference between the frequencies, $ \Delta f=f_2-f_1 $, is greater than the standard deviation of the estimation of $ \Delta f $}, shown as $ \Omega_f=\Delta f=\sqrt{{\rm CRB}(\Delta f)} $.
This means that \ac{srl} is obtained  by solving the following equation,
\begin{align}
\label{equ:CRB_fd}
{\rm CRB}(\Delta f)=\bm{u}^{T}{\rm CRB}(\bm{f})\bm{u},
\end{align}
where $ \bm{u}=[1,-1]^T $.
The work in \cite{el2011statistical} extends these criteria to multidimensional harmonic retrieval cases.

However, the above metrics are only feasible in scenarios with appropriately high \ac{snr} and no modeling or signal mismatch \cite{Smith_resolution05}. 
For example, when the \ac{snr} approaches infinity, the \ac{crb} decreases to 0, implying that the resolution approaches 0 by theses criteria.
This apparently contradicts the principle that the angular resolution of arrays mainly depends on the aperture, known as Rayleigh resolution limit \cite{Rayleigh}.
The contradiction mainly stems from the significant influence of noise on the absolute value of \ac{crb}, and the scope of application is hence limited.

Directly using the existing \ac{crb} based criteria to analyze the angular resolution limit of partly calibrated arrays leads to impractical conclusions: 
We take the Smith criterion as an example.
Denote the angular resolution limit of fully and partly calibrated arrays by $ \Omega_{\rm FC} $ and $ \Omega_{\rm PC} $, respectively.
Smith criterion in \eqref{equ:CRB_fd} implies that  
\begin{align}
\Omega_{\rm FC}^2&=\bm{u}^T{\rm CRB}_{\rm FC}(\bm{\omega})\bm{u}, \\
\Omega_{\rm PC}^2&=\bm{u}^T{\rm CRB}_{\rm PC}(\bm{\omega})\bm{u},
\end{align}
where $ \bm{u}=[1,-1]^T $. 
Since $ {\rm CRB}_{\rm PC}\succcurlyeq{\rm CRB}_{\rm FC} $ \cite{RARE02}, we have $ \Omega_{\rm PC}\geqslant\Omega_{\rm FC} $, yielding that the angular resolution of partly calibrated arrays can be much worse than fully calibrated arrays'.
This conclusion conflicts with the high-resolution performance of existing direction-finding algorithms for partly calibrated arrays \cite{RARE02,exRARE04,zgb23}.

In summary, the performance of existing \ac{crb} based criteria is seriously affected by noise and model error.
The above shortcomings motivate us to propose a new resolution criterion that is less sensitive to noise and explains the angular resolution of partly calibrated arrays more practically, which is detailed in Subsection~\ref{subsec:novel_usage_CRB}.

\section{Signal model}
\label{sec:signal}


Consider $ K $ linear antenna subarrays with $ N $ array elements in total located on a line. 
The $ k $-th subarray is composed of $ |\mathcal{N}_k| $ elements, where $ \mathcal{N}_k $ denotes the index set of elements in the $ k $-th subarray and the full index set is denoted by $\mathcal{N}\equiv \bigcup_{k}\mathcal{N}_k=\{1,\dots,N\} $.
Denote by $ \varphi_n $ the position of the $ n $-th element relative to the 1-st element's for $ n\in \mathcal{N} $.
Without loss of generality, we assume $ \varphi_N>\dots>\varphi_1=0 $, and denote the whole array aperture by $ D=\varphi_N $.
Partly calibrated arrays assume that the intra-subarray displacements, $ \varphi_{p_k}-\varphi_{q_k} $ with $ p_k,q_k\in\mathcal{N}_k $ are exactly known or well calibrated, however, the inter-subarray displacement between the $ k $-th subarray and the 1-st subarray, denoted by $ \xi_k\equiv\varphi_{\bar{n}_{k-1}+1}-\varphi_{1}=\varphi_{\bar{n}_{k-1}+1} $, is assumed  \emph{unknown}, where the number of elements in the first $ k $ subarrays is defined as $ \bar{n}_k=\sum_{i=1}^{k}|\mathcal{N}_{i}| $ for $ k=1,\dots,K $ ($ \xi_1=0 $).
In comparison, fully calibrated arrays assume that the positions of all the elements, $ \{\varphi_n\}_{n \in \mathcal{N}} $, are exactly known.
Let $ \bm{\xi}=[\xi_2,\dots,\xi_K]^T\in\mathbb{R}^{(K-1)\times1} $.    
Assume that the $ N $ elements share a common/well-synchronized sampling clock.
The diagram of partly calibrated array is shown in Fig.~\ref{fig:array}.

\begin{figure}[!htbp]
\centering 
\includegraphics[width=3.5in]{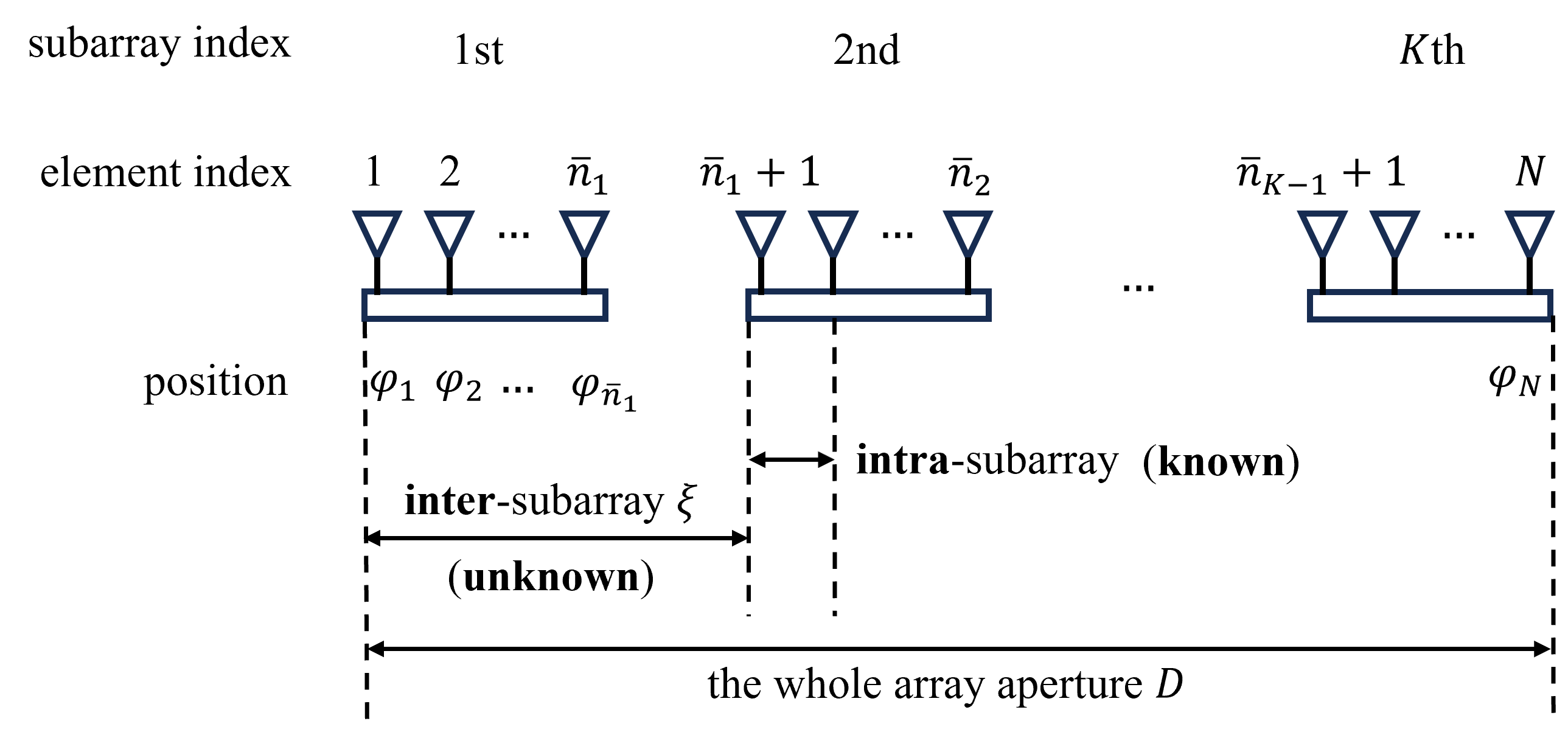}
\caption{Partly calibrated array.}
\label{fig:array}
\end{figure}

Consider $ L $ far-field \cite{far_field}, narrow-band sources impinging the signals onto the whole array from different directions $ \theta_1,\dots,\theta_L\in(-\pi/2,\pi/2) $ with $ \theta_1<\dots<\theta_L $.
Assume that these $ L $ sources are identifiable for each subarray, i.e., $ L<\min\{|\mathcal{N}_k|,k=1,\dots,K\} $.
Denote by $ \omega_l=\frac{2\pi\sin\theta_l}{\lambda} $ the spatial angular frequency of the $ l $-th source for $ l=1,\dots,L $, where $ \lambda $ is the wavelength of the source signals.
The Rayleigh resolution limit of $ \omega $ in the fully calibrated array case is defined as $ \Omega\equiv2\pi/D $ (ignoring the coefficient).
Denote the vectors of directions and spatial angular frequencies by $ \bm{\theta}=[\theta_1,\dots,\theta_L]^T\in\mathbb{R}^{L\times1} $ and $ \bm{\omega}=[\omega_1,\dots,\omega_L]^T\in\mathbb{R}^{L\times1} $, respectively. 
The maximum spatial frequency separation is defined as $ \Delta\omega\equiv\omega_L-\omega_1>0 $.

In the partly calibrated array case, the received signal of the $ k $-th subarray is expressed as
\begin{align}
\label{equ:yk}
\bm{y}_k(t)=\bm{A}_k(\bm{\omega},\xi)\bm{s}(t)+\bm{n}_k(t)\in\mathbb{C}^{|\mathcal{N}_k|\times1},
\end{align}
where $ \bm{A}_k(\bm{\omega},\xi)=[\bm{a}_k(\omega_1,\xi),\dots,\bm{a}_k(\omega_L,\xi)]\in\mathbb{C}^{|\mathcal{N}_k|\times L} $ is the $ k $-th steering matrix, $ [\bm{a}_k(\omega_l,\xi)]_{i}=e^{j\omega_l\varphi_n} $ for $ i=n-\bar{n}_{k-1} $ and $ n\in\mathcal{N}_k $, $ \bm{s}(t)\in\mathbb{C}^{L\times1} $ contains the complex coefficients of the sources, $ \bm{n}_k(t)\in\mathbb{C}^{|\mathcal{N}_k|\times1} $ denotes white noise with power $ \sigma^2 $, and $ t $ denotes sample time, $ t=t_1,t_2,\dots,t_T $.
Note that $ [\bm{a}_k(\omega_l,\xi)]_{i}=e^{j\omega_l(\varphi_n-\xi_k)}\cdot e^{j\omega_l\xi_k} $, where $ \varphi_n-\xi_k $ denotes the known intra-subarray displacement and $ \xi_k $ denotes the unknown inter-subarray displacement.

Stacking the $ K $ received signals in \eqref{equ:yk} together yields
\begin{align}
\label{equ:stack}
\bm{y}(t)=\bm{A}(\bm{\omega},\bm{\xi})\bm{s}(t)+\bm{n}(t),
\end{align}
where 
\begin{align}
\bm{y}(t)&=[\bm{y}_1^T(t),\dots,\bm{y}_K^T(t)]^T\in\mathbb{C}^{N\times1}, \nonumber \\ \bm{A}(\bm{\omega},\bm{\xi})&=[\bm{a}(\omega_1,\bm{\xi}),\dots,\bm{a}(\omega_L,\bm{\xi})]\in\mathbb{C}^{N\times L}, \nonumber \\ \bm{a}(\omega_l,\bm{\xi})&=[\bm{a}_1^T(\omega_l,\xi),\dots,\bm{a}_K^T(\omega_l,\xi)]^T\in\mathbb{C}^{N\times1}, \nonumber \\ \bm{n}(t)&=[\bm{n}_1^T(t),\dots,\bm{n}_K^T(t)]^T\in\mathbb{C}^{N\times1}. \nonumber
\end{align}
In \eqref{equ:stack}, $ \bm{y}(t) $ is known, $ \bm{A}(\bm{\omega},\bm{\xi}),\bm{s}(t),\bm{n}(t) $ are unknown and $ \bm{\omega} $ is to be estimated.
For fully and partly calibrated arrays, $ \bm{\xi} $ in $ \bm{A}(\bm{\omega},\bm{\xi}) $ is assumed to be completely known and unknown, respectively.
For conciseness, we abbreviate $ \bm{A}(\bm{\omega},\bm{\xi}) $ to $ \bm{A} $, $ \bm{a}_k(\omega_l,\xi) $ to $ \bm{a}_{k}(\omega_l) $ and $ \bm{a}(\omega_l,\bm{\xi}) $ to $ \bm{a}(\omega_l) $, respectively.

The \ac{crb}s of spatial frequencies $\bm{\omega}$ in fully and partly calibrated arrays, denoted by $ {\rm CRB}_{\rm FC}(\bm{\omega}) $ and $ {\rm CRB}_{\rm PC}(\bm{\omega}) $, respectively, are shown in Appendix~\ref{app:cal_CRB}. 





\section{Our main contributions}
\label{sec:indicate}

In this section, we detail our main contributions: 1) propose a new \ac{crb} based resolution criterion less sensitive to noise and model error;
2) use the proposed criterion to show that partly calibrated arrays achieve high angular resolution similar to fully calibrated arrays. 

Particularly, we first introduce the intuition behind the proposed criterion of indicating angular resolution by the turning point of CRB in Subsection~\ref{subsec:novel_usage_CRB}. 
We then analyze the angular resolution of partly calibrated arrays and show the high-resolution ability in Subsection~\ref{subsec:contribution_partly}.
Finally, we apply the proposed resolution criterion to fully calibrated arrays and compare the analysis results with prior works in Subsection~\ref{subsec:contribution_fully}.

\subsection{The proposed resolution criterion}
\label{subsec:novel_usage_CRB}

To meet the challenge of analyzing the resolution limit of partly calibrated arrays, we propose a new resolution criterion using the turning point of \ac{crb}. 
The proposed criterion states that: \emph{two signals are said to be resolvable w.r.t. the frequencies if the difference between the frequencies, $ \Delta f=f_2-f_1 $, is greater than the turning point of \ac{crb} w.r.t. $ \Delta f $ denoted by $ \Omega_f\equiv\mathcal{T}\left({\rm CRB}(\bm{f})\right) $.}
The definition of the \ac{crb} turning point is detailed in Section~\ref{subsec:turningpoint}.

The basic principle of this criterion is based on a significant phenomenon of $ {\rm CRB}(\bm{f}) $: Fix the source of $ f_1 $ as a reference and gradually increase the frequency difference $ \Delta f>0 $. The $ {\rm CRB}(\bm{f}) $ first declines rapidly w.r.t. $ \Delta f $, and then it almost remains constant (with a small fluctuation). The turning point is located close to the angular resolution limit.
We show this phenomenon in fully and partly calibrated arrays as an example.
Consider using $ 10 $ half-wavelength, uniform linear subarrays to estimate the directions of $ 2 $ sources with the angles $ \bm{\theta}=[1.2^{\circ},1.2^{\circ}+\Delta\theta]^T $, where $ \Delta\theta>0 $ denotes the angle difference.
There are $ 10 $ elements in each subarray and the adjacent subarrays are spaced at half-wavelength apart.
Denote by $ \Omega_f\equiv2\pi/D $ the angular resolution limit, where $ D $ is the whole array aperture.
The \ac{crb}s of fully and partly calibrated arrays, $ {\rm CRB}_{\rm FC} $ and $ {\rm CRB}_{\rm PC} $, are shown in Fig.~\ref{fig:show}. 

\begin{figure}[!htbp]
\centering 
\includegraphics[width=3.2in]{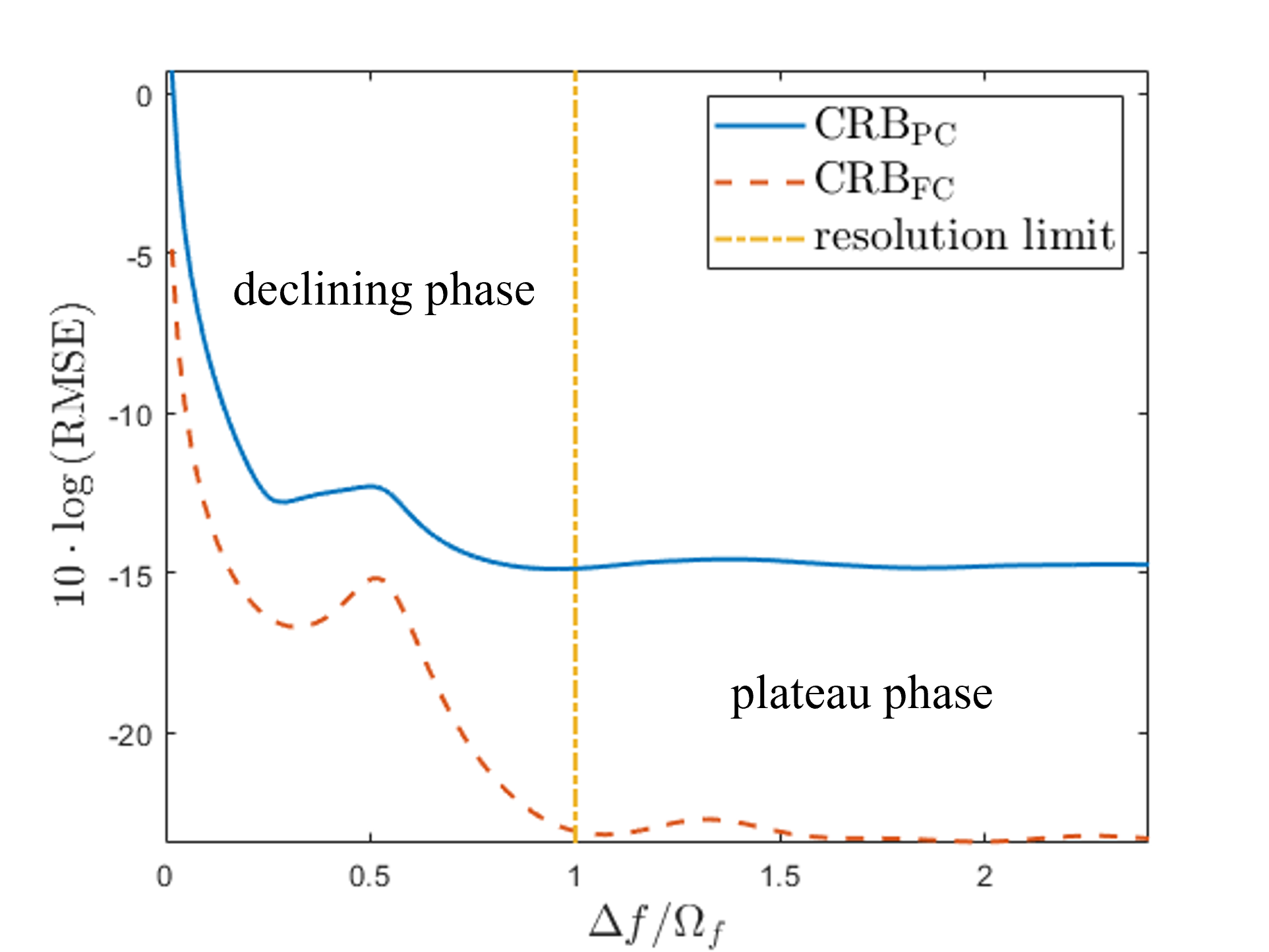}
\caption{$ {\rm CRB}_{\rm FC} $ and $ {\rm CRB}_{\rm PC} $ w.r.t. $ \Delta f $.}
\label{fig:show}
\end{figure}

From Fig.~\ref{fig:show}, we find that the two \ac{crb}s decline when $ \Delta f/\Omega_f\leqslant1 $, and tend to be stable when $ \Delta f/\Omega_f>1 $.
We provide an intuitive explanation of this phenomenon: 
when the frequency difference $ \Delta f $ is within the resolution limit $ \Omega_f $, the sources cannot be distinguished, disrupting the estimation performance and greatly increasing the \ac{crb}; 
in this case, increasing $ \Delta f $ increases the resolvability and hence decreases the \ac{crb}, which corresponds to the \emph{declining phase} of the \ac{crb}. 
When $ \Delta f $ exceeds $ \Omega_f $, the estimation accuracy is close to the counterpart of estimating the sources separately \cite{Lee92}; consequently, the accuracy is less relevant to $ \Delta f $, keeping the \ac{crb} constant,  which corresponds to the \emph{plateau phase} of the \ac{crb}.
Here, 'plateaus phase' refers to a situation where, after some change, a trend or curve stabilizes and no longer shows significant variation or fluctuation, which is consistent with the \ac{crb} when $ \Delta f\geqslant \Omega_f $.
Therefore, the turning point between the declining and the plateau phases indicates angular resolution.

For another reason, we use the turning point of \ac{crb} as the angular resolution limit since it can reflect the influence on each other for parameter estimation, which is precisely the meaning of resolution.
Particularly, when two sources are not distinguished, they significantly influence each other, resulting in poor estimation performance as reflected in the decline part of \ac{crb}. 
When the two sources can be separated, their mutual influence is minimal, and the CRB approaches the performance of separate estimation. Thus, the transition point between these two states can be used to indicate the resolution, which is the turning point of \ac{crb}.

The reason that the turning point is not exactly located at $ \Delta f=\Omega_f $ is that on one hand, the theoretical proof in the paper represents the result in a statistical average sense (see the A2 assumption in Section~\ref{sec:analysis}), which cannot ensure that the turning point of the \ac{crb} in every specific scenario precisely locates on the resolution cell; 
on the other hand, resolution criterion, either the Rayleigh resolution limit or 3dB beam width, is an empirical concept, which is not an absolute indication of separability or non-separability.
Overall, it can ensure that the magnitude of the resolution is correct.

In the proposed resolution criterion, angular resolution depends not only on the angular difference but also on the absolute angle values.
This is because that we use the frequency separation, $ \Delta f $, to reflect source separation instead of angle separation, $ \Delta\theta $.
Particularly in Fig.~\ref{fig:show}, we have $ \Delta f=2\pi(\sin\theta_2-\sin\theta_1)/\lambda $ and $ \Delta\theta=\theta_2-\theta_1 $.
Through trigonometric transformation, we have
\begin{align}
\Delta f=\frac{4\pi}{\lambda} \sin\left(\frac{\Delta\theta}{2}\right)\cos\left(\frac{\theta_2+\theta_1}{2}\right).
\end{align}
In the above equation, the sin part reflects the angular resolution and the cos part reflects the angle values themselves.
Therefore, the \ac{crb} is not only related to the angle separation $ \Delta\theta $, but also angle values $ (\theta_2+\theta_1)/2 $.
It can be found that large angle values have worse resolution.

Note that the proposed resolution criterion using the \ac{crb} turning point is almost unaffected by \ac{snr}.
This is because in the \ac{crb} expression (see \eqref{equ:CRB_FC} and \eqref{equ:CRB_PC} in Appendix~\ref{app:cal_CRB}), the component of white noise $ \sigma^2 $ can be isolated, which means that it only affects the absolute value of the \ac{crb} without altering its relative variation with respect to the angle difference. 
Although this property contradicts the commonly held conclusion that resolution is related to \ac{snr}, it is rational in analyzing the effect of position errors in distributed arrays.
This is because that inter-subarray position errors fundamentally differ from white noise errors:
the former are multiplicative errors, while the latter are additive errors.
When the \ac{snr} is extremely low, it becomes impossible to distinguish and estimate the angles, making it ineffective to analyze the impact of multiplicative errors on resolution. 
Therefore, we aim to conduct an analysis method that is unaffected or minimally affected by additive noise. The proposed resolution criterion using the \ac{crb} turning point can effectively address this challenge. 

In the sequel, we apply the proposed criterion to analyze the angular resolution of fully and partly calibrated arrays.
A main conclusion is that partly calibrated arrays achieve high resolution similar to that of fully calibrated arrays.
Note that the above phenomenon of \ac{crb} w.r.t. $ \Delta f $ was also mentioned in previous works \cite{Lee92,Smith_resolution05}, but they did not use it to indicate resolution and provide the corresponding theoretical guarantees.
The key challenge is to analyze the partial derivative of \ac{crb} to $ \Delta f $, which is hard to analytically calculate.
In Section~\ref{sec:analysis}, we provide an approximate method to solve this problem.


\subsection{Resolution analysis on partly calibrated arrays}
\label{subsec:contribution_partly}

Here, we analyze the declining phase, plateau phase, and turning point of $ {\rm CRB}_{\rm PC} $, and use the turning point to indicate angular resolution of partly calibrated arrays.

\subsubsection{Declining phase of \texorpdfstring{${\rm CRB}_{\rm PC} $}{e}}

We analyze the declining phase of $ {\rm CRB}_{\rm PC} $ using small quantity approximation, shown as Proposition~\ref{prop:prop2}, which illustrates that the main declining rate of $ {\rm CRB}_{\rm PC} $ is proportional to $ (\Delta\omega)^{-2(L-1)} $. 

\begin{proposition}[declining phase for partly calibrated arrays]
\label{prop:prop2}
When $ \Delta\omega\ll\Omega $, we have
\begin{align}
\label{equ:insideprop2}
{\rm CRB}_{\rm PC}=(\Delta\omega)^{-2(L-1)}\bm{C}_G+O((\Delta\omega)^{-2(L-1)+1}),
\end{align}
where $ \bm{C}_G $ is a constant w.r.t. $ \Delta\omega $.
\end{proposition}
\begin{proof}
See Appendix \ref{app:proof_of_prop2}.
\end{proof}

\subsubsection{Plateau phase of \texorpdfstring{${\rm CRB}_{\rm PC} $}{e}}

We analyze the plateau phase of $ {\rm CRB}_{\rm PC} $, shown as Proposition~\ref{prop:prop_CRB_approx2}, where the assumptions are detailed in Subsection~\ref{subsec:assumptions}.
Proposition~\ref{prop:prop_CRB_approx2} illustrates that $ {\rm CRB}_{\rm PC} $ remains almost constant w.r.t. $ \Delta\omega $ when $ \Delta\omega\geqslant\Omega $. 
Here $ {\rm CRB}(\Delta\omega) $ denotes the functional relationship between \ac{crb} and $ \Delta\omega $ instead of the \ac{crb} of $ \Delta\omega $.

\begin{proposition}[plateau phase for partly calibrated arrays]
\label{prop:prop_CRB_approx2}
Under assumptions A1-A3, when $ \Delta\omega\geqslant\Omega $, we have 
\begin{align}
\label{equ:approx_now2}
\left|\frac{\partial {\rm CRB}_{\rm PC}(\Delta\omega)}{\partial \Delta\omega}\right|\approx 0
\end{align}
with high probability.
\end{proposition}
\begin{proof}
See Subsection \ref{subsec:the_proof1}.
\end{proof}

This proposition illustrates that under the scenario of using a large, sparse, uniformly distributed array to resolve two closely spaced sources (corresponding to assumptions A1-A3 in Section \ref{subsec:assumptions}), $ {\rm CRB}_{\rm PC} $ remains almost constant w.r.t. $ \Delta\omega $ when $ \Delta\omega\geqslant\Omega $.
Based on the proposed criterion, the stable \ac{crb} implies that the sources are resolvable in this situation.
The approximation in Proposition~\ref{prop:prop_CRB_approx2} is used for the convenience of explanation and a more rigorous expression is detailed in the corresponding proof.

\subsubsection{Turning point of \texorpdfstring{${\rm CRB}_{\rm PC} $}{e}}

Intuitively, the intersection of the declining phase and plateau phase of \ac{crb} corresponds to the turning point.
Based on Proposition~\ref{prop:prop2} and Proposition~\ref{prop:prop_CRB_approx2}, a criterion \cite{fishler2006spatial} is used to determine the turning point of $ {\rm CRB}_{\rm PC} $ in a strict sense, given by
\begin{align}
\label{equ:CRB_PC_turningpoint}
\mathcal{T}\left({\rm CRB}_{\rm PC}\right)=\Omega=2\pi/D,
\end{align}
where the explanation is detailed in Subsection~\ref{subsec:turningpoint}.

The conclusion of applying our proposed resolution criterion to $ {\rm CRB}_{\rm PC} $ in \eqref{equ:CRB_PC_turningpoint} implies that high angular resolution is achievable for partly calibrated arrays, inversely proportional to the whole array aperture $ D $, which is consistent with existing direction-finding algorithms for partly calibrated arrays \cite{RARE02,exRARE04,zgb23}.
A main advantage relative to the existing \ac{crb} based \ac{srl} criteria is that the proposed criterion is less sensitive to noise.
This helps to focus on the main factors that affect the resolution (array aperture), while reducing the interference of secondary factors (noise), and avoid the limitations of existing \ac{crb} based criteria.


\subsection{Resolution analysis on fully calibrated arrays}
\label{subsec:contribution_fully}

For comparison, we analyze the declining phase, plateau phase, and turning point of $ {\rm CRB}_{\rm FC} $, and use the turning point to indicate angular resolution of fully calibrated arrays.

\subsubsection{Declining phase of \texorpdfstring{${\rm CRB}_{\rm FC} $}{e}}

The declining phase of $ {\rm CRB}_{\rm FC} $ has been theoretically analyzed using small quantity approximation in \cite{Lee92}, shown as Lemma~\ref{prop:prop1}.
It is proven that when $ \Delta\omega\ll\Omega $, $ {\rm CRB}_{\rm FC} $ declines at the rate mainly proportional to $ (\Delta\omega)^{-2(L-1)} $, which is the same as the counterpart of $ {\rm CRB}_{\rm PC} $.

\begin{lemma}[declining phase for fully calibrated arrays\cite{Lee92}]
\label{prop:prop1}
When $ \Delta\omega\ll\Omega $, we have
\begin{align}
\label{equ:LeeF}
{\rm CRB}_{\rm FC}=(\Delta\omega)^{-2(L-1)}\bm{C}_F+O((\Delta\omega)^{-2(L-1)+1}),
\end{align}
where $ \bm{C}_{F} $ is a constant w.r.t. $ \Delta\omega $.
\end{lemma}



\subsubsection{Plateau phase of \texorpdfstring{${\rm CRB}_{\rm FC} $}{e}}

We analyze the plateau phase of $ {\rm CRB}_{\rm FC} $ in Proposition~\ref{prop:prop_CRB_approx1}, which illustrates that $ {\rm CRB}_{\rm FC} $ remains almost constant w.r.t. $ \Delta\omega $ when $ \Delta\omega\geqslant\Omega $.
The rigorous expression is similar to the counterpart of $ {\rm CRB}_{\rm PC} $, which is detailed in the corresponding proof.

\begin{proposition}[plateau phase for fully calibrated arrays]
\label{prop:prop_CRB_approx1}
Under assumptions A1-A3, when $ \Delta\omega\geqslant\Omega $, we have 
\begin{align}
\label{equ:approx_now1}
\left|\frac{\partial {\rm CRB}_{\rm FC}(\Delta\omega)}{\partial \Delta\omega}\right|\approx 0
\end{align}
with high probability.
\end{proposition}
\begin{proof}
See Subsection \ref{subsec:the_proof2}.
\end{proof}

\subsubsection{Turning point of \texorpdfstring{${\rm CRB}_{\rm FC} $}{e}}

Based on Lemma~\ref{prop:prop1} and Proposition~\ref{prop:prop_CRB_approx1}, we use the criterion in \cite{fishler2006spatial} to determine the turning point of $ {\rm CRB}_{\rm FC} $ in a strict sense, given by
\begin{align}
\label{equ:CRB_FC_turningpoint}
\mathcal{T}\left({\rm CRB}_{\rm FC}\right)=\Omega=2\pi/D,
\end{align}
where the explanation is similar with that of $ {\rm CRB}_{\rm PC} $.

Based on the proposed criterion, \eqref{equ:CRB_FC_turningpoint} means that the resolution limit of fully calibrated arrays is inversely proportional to the whole array aperture $ D $, which is consistent with existing Rayleigh resolution limit.
This verifies the feasibility of the proposed criterion and supports to apply it on the resolution analysis of partly calibrated arrays.


Our main contributions on theoretically analyzing the declining phase, plateau phase and turning point of \ac{crb} are summarized in Table~\ref{Table:table1}.

\begin{table}[!htbp]
\caption{Our main contributions}
\label{Table:table1}
\renewcommand\arraystretch{1.5}
\centering
\begin{tabular}{c|c|c}
\hline
 & \textbf{fully calibrated} & \textbf{partly calibrated} \\
\hline
\textbf{declining phase} & Lemma~\ref{prop:prop1} \cite{Lee92} & Proposition~\ref{prop:prop2} \\
\hline
\textbf{plateau phase} & Proposition~\ref{prop:prop_CRB_approx1} & Proposition~\ref{prop:prop_CRB_approx2} \\
\hline
\textbf{turning point} & $ \Omega=2\pi/D $ & $ \Omega=2\pi/D $ \\
\hline
\end{tabular}
\end{table}

Note that the super-resolution phenomenon in the existing self-calibration methods does not contradict the conclusion of this paper.
This is because existing super-resolution algorithms \cite{malioutov2005sparse,donoho2006compressed} make additional prior assumptions about the scenario and model, whereas the signal model in this paper does not. 
For instance, sparse recovery algorithms assume that targets are sparsely located within the solution space, and subspace methods assume uncorrelated source signals. 
These assumptions introduce extra feature information compared to the classical model, thereby affecting the model's performance bounds, which manifest as improvements in resolution.
However, we base our study on the classical model assumptions without incorporating other prior assumptions such as sparsity, and therefore, it does not involve super-resolution performance.
 


\section{Proof of the propositions in Section~\ref{sec:indicate}}
\label{sec:analysis}

In this section, we prove Propositions~\ref{prop:prop_CRB_approx2} and \ref{prop:prop_CRB_approx1}, while the proof of Proposition \ref{prop:prop2} is left to Appendix \ref{app:proof_of_prop2} since it is a direct extension of \cite{Lee92}. 
First, we introduce assumptions A1-A3 in Subsection \ref{subsec:assumptions}, followed by the detailed proofs of Propositions~\ref{prop:prop_CRB_approx2} and \ref{prop:prop_CRB_approx1} in Subsection \ref{subsec:the_proof1} and \ref{subsec:the_proof2}, respectively.
Based on the  proofs above, we explain how to determine the turning point of \ac{crb} in Subsection \ref{subsec:turningpoint}.

\subsection{Assumptions}
\label{subsec:assumptions}

To analyze the angular resolution in fully and partly calibrated arrays, we impose the following assumptions (A3 is not necessary for fully  calibrated arrays.). A diagram is shown in Fig.~\ref{fig:assumption}.
\begin{itemize}
    \item[A1:] Consider $ L=2 $ sources with spatial frequencies denoted by $ \omega_1=\omega_0 $ and $ \omega_2=\omega_0+\Delta\omega $.
    \item[A2:] The average of the element positions of each subarray is uniformly distributed in $ [0,D] $, i.e., $ \bar{\varphi}_k\sim\mathcal{U}[0,D] $, where 
    \begin{align}
    \label{equ:varphik}
    \bar{\varphi}_k= \frac{\sum_{n\in\mathcal{N}_k}\varphi_n}{|\mathcal{N}_k|}.
    \end{align}
    Each subarray has the same number of elements, i.e., $ |\mathcal{N}_k|=\frac{N}{K} $, and the number of subarrays, $ K $, is large such that $ 1/K\approx 0 $. 
    \item[A3:] The interval between the array elements within a subarray is small relative to the whole distributed array, i.e., a subarray can be approximated as a point in the geometry.
    Particularly, assume $ \varphi_n\approx\bar{\varphi}_k $ for $ n\in\mathcal{N}_k $, such that
    \begin{align}
    \label{equ:average}
    \frac{1}{|\mathcal{N}_k|}\sum_{n\in\mathcal{N}_k} f(\varphi_n)\cdot g(e^{j\Delta\omega\varphi_n})\approx f(\bar{\varphi}_k)\cdot g(Q_0^k),
    \end{align}
    where 
    \begin{align}
    \label{equ:Q0k}
    Q_0^k=\frac{\sum_{n\in\mathcal{N}_k} e^{j\Delta\omega\varphi_n}}{|\mathcal{N}_k|},
    \end{align}
    and $ f(\cdot), g(\cdot) $ are any general polynomial functions.
\end{itemize}

\begin{figure}[!htbp]
\centering 
\includegraphics[width=3.2in]{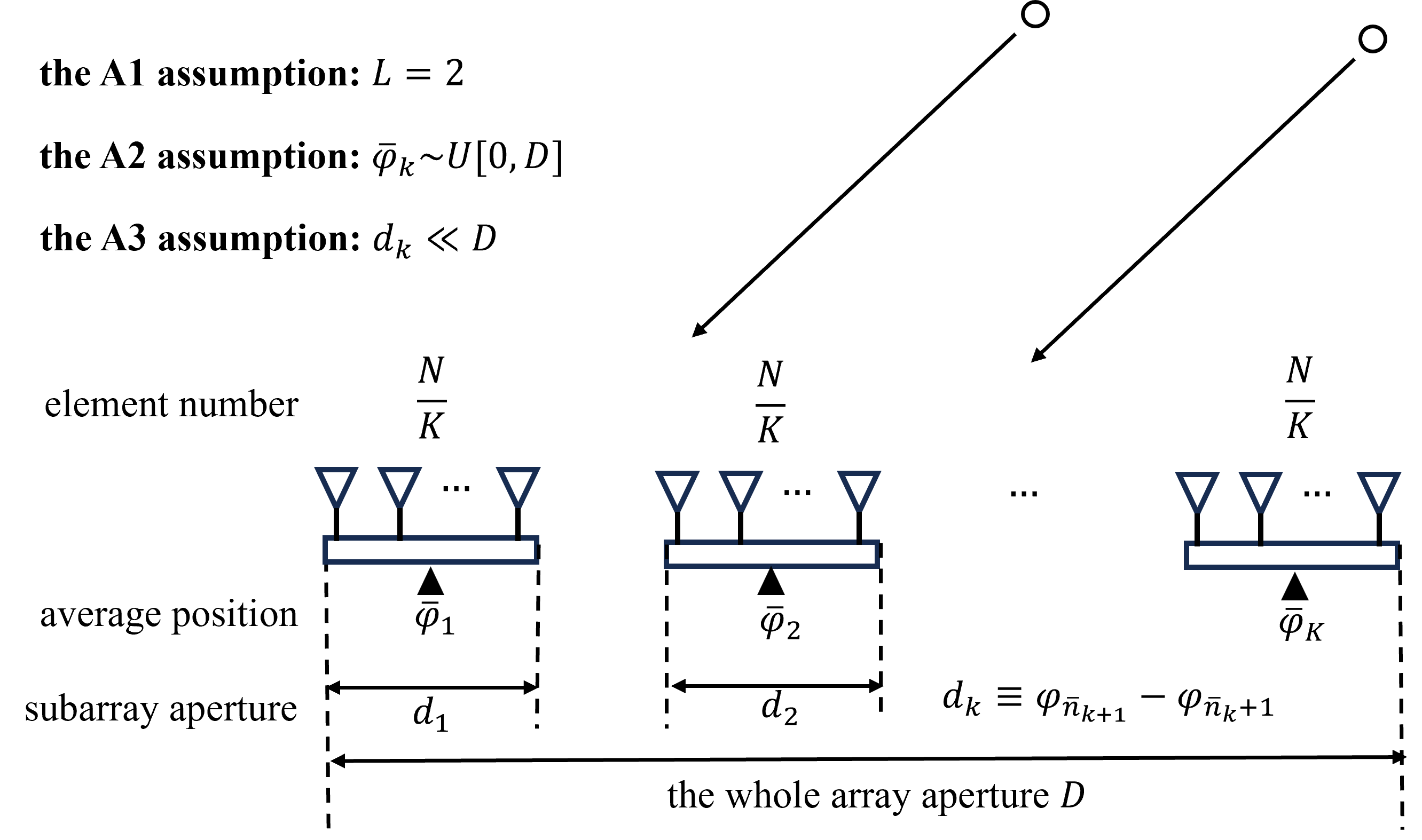}
\caption{Assumptions A1-A3 in partly calibrated arrays.}
\label{fig:assumption}
\end{figure}

In Assumption~A1, we mainly consider the case $ L=2 $ to simplify the expressions. 
In Assumption A2, uniformly distributed subarrays are common in practice.
Note that Assumption 2 can be extended to more general cases, such as the case where the inter-subarray position errors are bounded instead of being distributed on the whole array aperture.
These proofs are similar and we take Assumption A2 for example in this paper.
In Assumption A3, we consider that the apertures of the subarrays, $ d_k\equiv\varphi_{\bar{n}_{k+1}}-\varphi_{\bar{n}_{k}+1}, k=1,\dots,K $, are far less than the whole distributed array aperture $ D $, such that a subarray can be viewed as a point geometrically, i.e., $ 
\varphi_n\approx\bar{\varphi}_k $ for $ n\in\mathcal{N}_k $.
In an extreme case, substituting $ \varphi_n=\bar{\varphi}_k $ for $ n\in\mathcal{N}_k $ to \eqref{equ:average} yields the corresponding equation.
Therefore, closer intra-subarray displacements between $ \varphi_n, n\in\mathcal{N}_k $ implies smaller approximation errors in \eqref{equ:average}.
This means that assumption A3 actually corresponds to a large, sparsely distributed array sensor network, which is usually used for high angular resolution.

\subsection{Proof of Proposition~\ref{prop:prop_CRB_approx2}}
\label{subsec:the_proof1}

The proof of Proposition~\ref{prop:prop_CRB_approx2} is not direct due to the complex form of $ {\rm CRB}_{\rm PC} $ w.r.t. $ \Delta\omega $.
We then introduce intermediate variables w.r.t. $ 
\Delta\omega $, denoted by $ Q(\Delta\omega) $. We divide the proof of Proposition~\ref{prop:prop_CRB_approx2} into several tractable lemmas.
In the sequel, we first introduce the intermediate variables $ Q $ and the lemmas, and then show how Proposition~\ref{prop:prop_CRB_approx2} is proved with these lemmas.

The intermediate variables $ Q $ are defined as
\begin{align}
\label{equ:threeQ}
Q_{i}(\Delta\omega)&=\sum_{n=1}^N \varphi_n^{i}\cdot e^{j\Delta\omega\varphi_n}\bigg/\sum_{n=1}^N \varphi_n^i, \\
\label{equ:threeQk}
Q_{i}^k(\Delta\omega)&=\sum_{n\in\mathcal{N}_k} \varphi_n^{i}\cdot e^{j\Delta\omega\varphi_n}\bigg/\sum_{n\in\mathcal{N}_k} \varphi_n^i,
\end{align}
where $ i=0,1,2 $ and $ k=1,\dots,K $. 
In the sequel, we abbreviate $ Q_{i}(\Delta\omega) $ and $ Q_{i}^k(\Delta\omega) $ as $ Q_{i} $ and $ Q_{i}^k $, respectively.
The intermediate variables $ Q_i $ correspond to the whole distributed array, and $ Q_i^k $ correspond to the $ k $-th subarray.
Our theoretical results are mainly about $ Q_i $, while $ Q_i^k$ are used for intermediate derivations.
These intermediate variables are all bounded by 1. Abandoning extreme scenarios ($ e^{j\Delta\omega\varphi_n}=1,\ n=1,\dots,N $), we assume
\begin{align}
\label{equ:Q<1}
|Q_i|<1,\ |Q_i^k|<1.
\end{align}
}

With the help of $ Q $, we rewrite the \ac{crb} w.r.t.  $ Q $ instead of $\Delta\omega$, facilitating the analysis of the plateau phase.
This is feasible because the proof of Lemma \ref{lemma:Q} and \ref{lemma:PC} is equivalent to that of Proposition~\ref{prop:prop_CRB_approx2}:
\begin{align}
{\rm Proposition}\ \ref{prop:prop_CRB_approx2} \Longleftrightarrow {\rm Lemma}\ \ref{lemma:Q}+{\rm Lemma}\ \ref{lemma:PC}, \nonumber
\end{align}

\noindent and the lemmas are shown as follows.
Lemma \ref{lemma:Q} shows that the intermediate variables $ Q $ and their derivatives w.r.t. $ \Delta\omega $ tend to be zero when $ \Delta\omega\geqslant\Omega $.
The approximation to zero in Lemma \ref{lemma:Q} is a rough but intuitive expression, and the more rigorous expression is shown in the proof. 
Lemma~\ref{lemma:PC} means that the main influence of $ \Delta\omega $ on $ 
{\rm CRB}_{\rm PC} $ is embodied by $ Q(\Delta\omega) $, which supports the analysis of how $ Q(\Delta\omega) $ affects $ {\rm CRB}_{\rm PC} $ instead.

\begin{lemma}
\label{lemma:Q}
Under assumptions A1-A3, when $ \Delta\omega\geqslant\Omega $, we have
\begin{align}
\label{equ:lemma_Q}
\left|Q(\Delta\omega)\right|\approx 0, \ \left|\frac{\partial Q(\Delta\omega)}{\partial \Delta\omega}\right|\approx 0.
\end{align}
\end{lemma}
\begin{proof}
See Appendix \ref{app:prop_4_1}.
\end{proof}

Here we give an intuitive explanation of how $ |Q(\Delta\omega)| $ is close to 0 in Lemma~\ref{lemma:Q}. 
From Lemma 2.~\ref{lemma:lemmaEQ} in Appendix~\ref{app:prop_4_1}, when $ t=0.25 $, $ \Delta\omega D=2\pi $ and $ N=200 $, we have 
\begin{align}
P\left(|Q_0|\geqslant 0.35\right)\leqslant 0.0077.
\end{align} 
The above conclusion can be extended to the cases of $ \Delta\omega D\geqslant 2\pi $.
Therefore, when $ N $ is large enough, we have $ |Q(\Delta\omega)|\approx 0 $ with high probability, yielding that Proposition~\ref{prop:prop_CRB_approx2} also holds with high probability.

\begin{lemma}
\label{lemma:PC}
Under assumptions A1-A3, when $ \Delta\omega\geqslant\Omega $, we have
\begin{align}
\label{equ:lemma_PC}
{\rm CRB}_{\rm PC}(\Delta\omega)\approx {\rm CRB}_{\rm PC}(Q(p\Delta\omega)),
\end{align}
where $ |p|\geqslant 1, p\in\mathbb{Z} $. 
\end{lemma}
\begin{proof}
See Appendix \ref{app:proof_of_prop4}.
\end{proof}

Based on the chain rule of partial derivatives, we have
\begin{align}
\label{equ:chain_rule}
\frac{\partial {\rm CRB}_{\rm PC}(\Delta\omega)}{\partial \Delta\omega}=\frac{\partial {\rm CRB}_{\rm PC}(\Delta\omega)}{\partial Q(p\Delta\omega)}\cdot\frac{\partial Q(p\Delta\omega)}{\partial \Delta\omega}.
\end{align}
Substituting $ {\rm CRB}_{\rm PC}(\Delta\omega)\approx {\rm CRB}_{\rm PC}(Q(p\Delta\omega)) $ in Lemma \ref{lemma:PC} to \eqref{equ:chain_rule} yields that 
\begin{align}
\label{equ:chain_rule_recast}
\frac{\partial {\rm CRB}_{\rm PC}(\Delta\omega)}{\partial \Delta\omega}\approx\frac{\partial {\rm CRB}_{\rm PC}(Q(p\Delta\omega))}{\partial Q(p\Delta\omega)}\cdot p\frac{\partial Q(p\Delta\omega)}{\partial p\Delta\omega}.
\end{align}
Under assumptions A1-A3, the \ac{fim} of \ac{crb} is not singular, yielding $ \left|\frac{\partial {\rm CRB}_{\rm PC}(Q(p\Delta\omega))}{\partial Q(p\Delta\omega)}\right|\leqslant C $, where $ C $ is some constant.
Based on Lemma~\ref{lemma:Q}, we have $ \left|\partial Q(p\Delta\omega)/\partial p\Delta\omega \right|\approx 0 $ when $ \Delta\omega\geqslant\Omega/p $, which is also satisfied for $ \Delta\omega\geqslant\Omega $ since $ |p|\geqslant 1 $.
Therefore, when $ \Delta\omega\geqslant\Omega $ we have
\begin{align}
\label{equ:lemma1_final}
\left|\frac{\partial {\rm CRB}_{\rm PC}(\Delta\omega)}{\partial \Delta\omega}\right|\approx 0,
\end{align}
completing the proof.

\subsection{Proof of Proposition~\ref{prop:prop_CRB_approx1}}
\label{subsec:the_proof2}

The proof of Proposition~\ref{prop:prop_CRB_approx1} is similar to the counterpart of Proposition~\ref{prop:prop_CRB_approx2}, except that Lemma~\ref{lemma:PC} is replaced by Lemma~\ref{lemma:FC}, given by:
\begin{align}
{\rm Proposition}\ \ref{prop:prop_CRB_approx1} \Longleftrightarrow {\rm Lemma}~\ref{lemma:Q}+{\rm Lemma}~\ref{lemma:FC}, \nonumber
\end{align}

\noindent where Lemma \ref{lemma:FC} demonstrates that the main influence of $ 
\Delta\omega $ on $ {\rm CRB}_{\rm FC} $ is embodied by $ Q(\Delta\omega) $:

\begin{lemma}
\label{lemma:FC}
Under assumptions A1-A3, when $ \Delta\omega\geqslant\Omega $, we have
\begin{align}
\label{equ:lemma_FC}
{\rm CRB}_{\rm FC}(\Delta\omega)= {\rm CRB}_{\rm FC}(Q(\Delta\omega)).
\end{align}
\end{lemma}
\begin{proof}
See Appendix \ref{app:proof_of_prop3}.
\end{proof}

Combining Lemma \ref{lemma:Q} and \ref{lemma:FC} yields Proposition~\ref{prop:prop_CRB_approx1}.
The proof is the same as the counterpart of Proposition~\ref{prop:prop_CRB_approx2} and is thus omitted here.

\subsection{Determining the turning point of CRB}
\label{subsec:turningpoint}

We explain how to determine the turning of \ac{crb}. 
From the above propositions, we know that when $ \Delta\omega\ll\Omega $, $ {\rm CRB}_{\rm FC} $ and $ {\rm CRB}_{\rm PC} $ decline rapidly as $ \Delta\omega $ increases, and plateau out when $ \Delta\omega\geqslant\Omega $.
A criterion is in demand to distinguish the plateau phase and the declining phase in a strict sense, which implies the turning point of \ac{crb}.

As the \ac{crb} curve has the identical trend with the intermediate variables $ Q(\Delta\omega) $, by observing the structure of $ Q(\Delta\omega) $ defined in  \eqref{equ:threeQ}, we use the criterion 
\begin{equation}
    \Delta\omega \varphi_N > 2\pi
\end{equation}
to determine the turning point of \ac{crb}. Consequently, the turning point is located at 
\begin{equation}
\label{equ:criterion}
    \Delta\omega=2\pi/\varphi_N=\Omega. 
\end{equation}



This criterion is inspired by \cite{fishler2006spatial}, which considers a similar problem that distinguishes the correlated and uncorrelated signals, detailed as follows:
Denote the correlation of two signals by $ E(\Delta f) $, 
\begin{align}
\label{equ:Eij}
E(\Delta f)=\frac{1}{\Delta x}\int_{-\frac{\Delta x}{2}}^ {\frac{\Delta x}{2}} e^{-j \Delta fx}dx,
\end{align}
where $ \Delta f $ is the frequency difference between two signals.
The signals are regarded as correlated if $ E(\Delta f) $  is close to 1 or regarded as uncorrelated if $ E(\Delta f) $ approaches 0.
In \cite{fishler2006spatial}, it is explained that a sufficient condition for $ |E(\Delta f)| $ to be much less than one is that the integrand completes at least one cycle or equivalently 
\begin{align}
\label{equ:Eij_criterion}
\Delta f\Delta x> 2\pi.
\end{align}
We use the similarity between $ Q(\Delta\omega) $ in \eqref{equ:threeQ} and $ E(\Delta f) $ in \eqref{equ:Eij} and determine the turning point of \ac{crb} as \eqref{equ:criterion}.

Finally, we give an intuitive explanation of this criterion applied in $ Q(\Delta\omega) $.
Take $ Q_0=\sum_{n=1}^N e^{j\Delta\omega\varphi_n}/N $ as an example.
When $ \Delta\omega\ll\Omega $, we have $ \Delta\omega\varphi_N\ll 2\pi $, which means that the phases of $ e^{j\Delta\omega\varphi_n} $ are centralized in a small range of $ [0,2\pi) $, yielding large $ |Q_0| $.
When $ \Delta\omega\geqslant\Omega $, we have $ \Delta\omega\varphi_N\geqslant 2\pi $, which means that the phases of $ e^{j\Delta\omega\varphi_n} $ are distributed in $ [0,2\pi) $ and vectors in different directions cancel each other, yielding small $ |Q_0| $.
A diagram w.r.t. the phases of $ e^{j\Delta\omega\varphi_n} $ and $ Q_0 $ is shown in Fig.~\ref{fig:phasse}
Since $ Q_1,Q_2 $ are the weighted extension of $ Q_0 $, they also approximately have the above characteristics.

\begin{figure}[!htbp]
\centering
\subfigure[$ \Delta\omega\ll\Omega $.]{
\label{subfigure:phasea}
\includegraphics[width=1.6in]{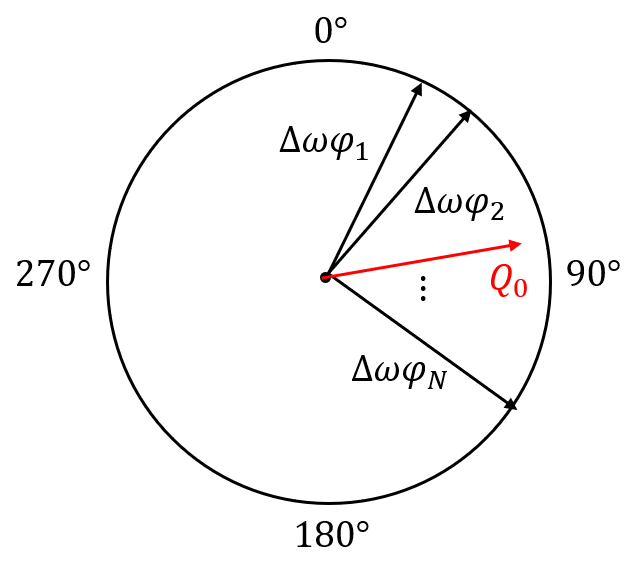}}
\subfigure[$ \Delta\omega\geqslant\Omega $.]{
\label{subfigure:phaseb}
\includegraphics[width=1.6in]{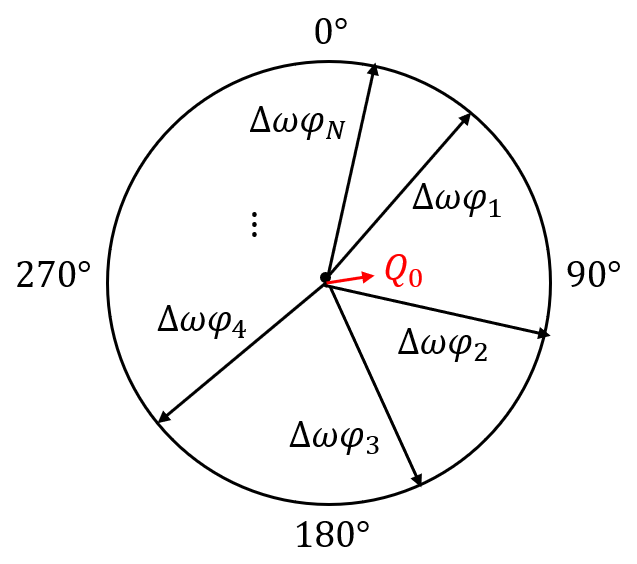}}
\caption{The phases of $ e^{j\Delta\omega\varphi_n} $ and $ Q_0 $.}
\label{fig:phasse}
\end{figure}

\section{Simulations}
\label{sec:simulation}

In this section, we first show the declining and plateau phases of $ {\rm CRB}_{\rm FC} $ and $ {\rm CRB}_{\rm PC} $ by the simulation results in Subsection \ref{subsec:example}, supporting the \ac{crb} analysis in Section~\ref{sec:indicate}.
We present that the turning points of $ {\rm CRB}_{\rm FC} $ and $ {\rm CRB}_{\rm PC} $ are not sensitive to \ac{snr} in Subsection~\ref{subsec:snr_sensitive}.
We then give the approximation errors in the proof of Lemma~\ref{lemma:PC} in Subsection \ref{subsec:approximation_all} to verify the feasibility of the approximation.
Finally, we explain that high angular resolution is achievable for both fully and partly calibrated arrays by subspace based algorithms in Subsection~\ref{subsec:algorithm}, verifying our main conclusion.

\subsection{Verification of the CRB analysis in Section~\ref{sec:indicate}}
\label{subsec:example}

We show the declining and plateau phases of $ {\rm CRB}_{\rm FC} $ and $ {\rm CRB}_{\rm PC} $ w.r.t. $ \Delta\omega $ by simulations to verify the theoretical analysis of \ac{crb} in Section~\ref{sec:indicate}.

Consider $ K=10 $ half-wavelength uniform linear subarrays, each composed of 10 elements, yielding $ N=100 $. These subarrays are uniformly spaced on a straight line with $ \xi_k=I(k-1)\lambda $ for $ k=1,\dots,K $, where $ I>0 $ reflects the size of interval between subarrays and is set as $ I=50 $.
The wavelength of the received signals is $ \lambda=1m $, and the resolution limit of $ \bm{\omega} $ is thus $ \Omega=0.014{\rm m}^{-1} $.
The geometry of the distributed array is shown in Fig.~\ref{fig:geometry}.

\begin{figure}[!htbp]
\centering
\includegraphics[width=3.2in]{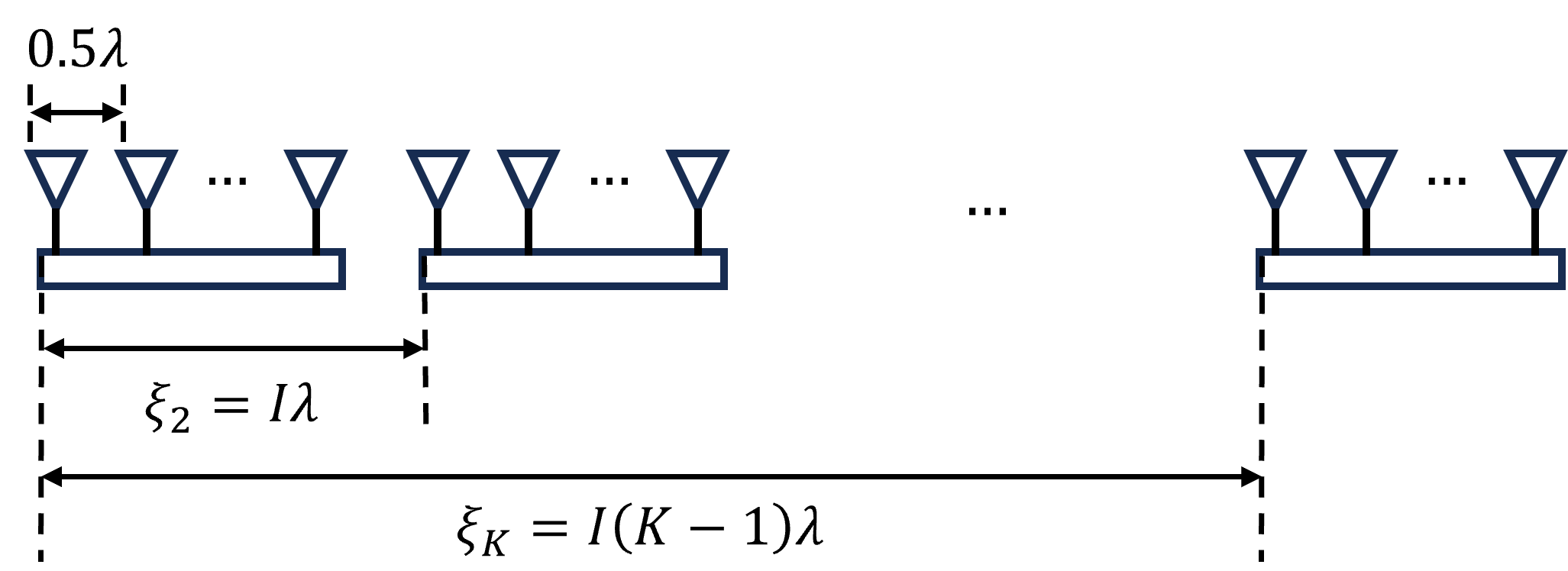}
\caption{The geometry of distributed arrays.}
\label{fig:geometry}
\end{figure}

The directions $ \bm{\theta} $ are uniformly in the range of $ [\theta_{\min},\theta_{\max}] $, i.e., $ \theta_l=\theta_{\min}+(l-1)(\theta_{\max}-\theta_{\min})/(L-1) $ for $ l=1,\dots,L $.
In this case, we have $ \Delta\omega=2\pi(\sin(\theta_{\max})-\sin(\theta_{\min}))/\lambda $, and $ \Delta\omega/(L-1) $ denotes the minimum separation between sources.
We set $ \theta_{\min}=1.2^{\circ} $.
The complex coefficients $ \bm{s} $ are set as $ s_l=e^{j\pi/5} $ for $ l=1,\dots,L $.
The \ac{snr} is defined as $ 1/\sigma^2 $ being 20dB.
We consider $ L=2,3,4 $, and show $ {\rm CRB}_{\rm FC} $ and $ {\rm CRB}_{\rm PC} $ w.r.t. $ (\Delta\omega/(L-1))/\Omega $ with a logarithmic coordinate in Fig.~\ref{fig:example}.


\begin{figure}[!htbp]
\centering
\includegraphics[width=3.2in]{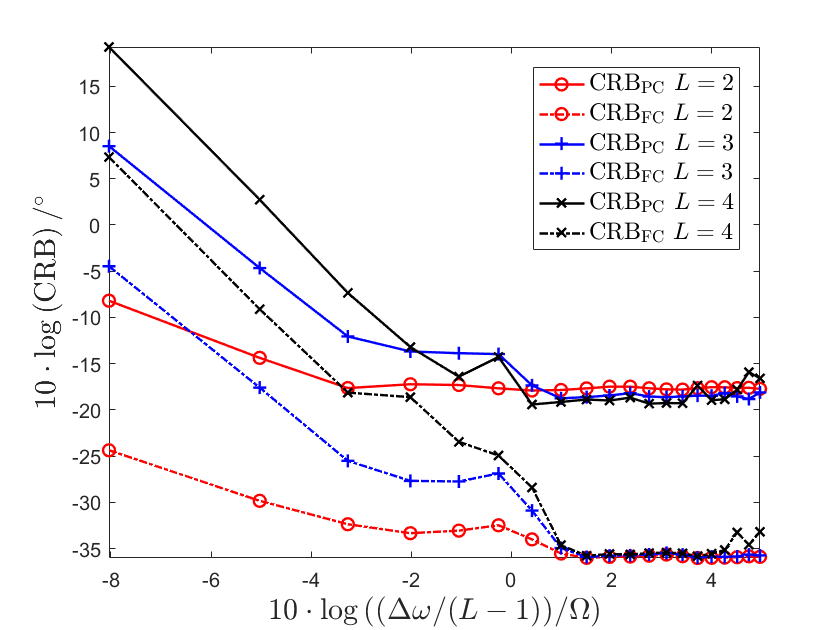}
\caption{$ {\rm CRB}_{\rm FC} $ and $ {\rm CRB}_{\rm PC} $ w.r.t. $ (\Delta\omega/(L-1))/\Omega $ for $ L=2,3,4 $.}
\label{fig:example}
\end{figure}

From Fig.~\ref{fig:example}, when $ \Delta\omega/(L-1)\ll\Omega $, we find that the slopes of $ {\rm CRB}_{\rm FC} $ and $ {\rm CRB}_{\rm PC} $ w.r.t. $ \Delta\omega $ in the $ L=2,3,4 $ cases are close to $ -2,-4,-6 $ in the logarithmic coordinate, respectively.
This verifies the conclusions in Lemma~\ref{prop:prop1} and Proposition~\ref{prop:prop2} that $ {\rm CRB}_{\rm FC} $ and $ {\rm CRB}_{\rm PC} $ are mainly proportional to $ (\Delta\omega)^{-2(L-1)} $ when $ \Delta\omega\ll\Omega $, which correspond to the declining phase of \ac{crb}.
When $ \Delta\omega/(L-1)\geqslant\Omega $, we find that $ {\rm CRB}_{\rm FC} $ and $ {\rm CRB}_{\rm PC} $ both begin to plateau out w.r.t. $ \Delta\omega $.
This verifies the conclusions in Proposition~\ref{prop:prop_CRB_approx2} and Proposition~\ref{prop:prop_CRB_approx1}, which correspond to the plateau phase of \ac{crb}.

Note that there is a fluctuation of \ac{crb} near the turning point, which is a common phenomenon in the \ac{crb} based resolution criterion.
An intuitive reason is that the \ac{crb} involving the inverse of matrix usually has a complex form w.r.t. $ \Delta\omega $, particularly when $ \Delta\omega $ is close to the resolution limit $ \Omega $.
Based on the theoretical analysis method proposed in this paper, we can also give a more convincing explanation on this phenomenon. 
Particularly, from the discussion in Subsection~\ref{subsec:the_proof1} and Subsection~\ref{subsec:the_proof2}, we transform the analysis of $ \partial {\rm CRB}(\Delta\omega)/\partial \Delta\omega $ into that of $ \partial Q(\Delta\omega)/\partial \Delta\omega $.
We take $ Q_0 $ as an example and statistically analyze the how $ Q_0(\Delta\omega) $ varies with $ \Delta\omega $.
Based on the A1-A3 assumptions, we calculate the expectations of $ Q_0 $ as follows.

\begin{align}
\label{equ:EQ0}
|\mathbb{E}_{Q_0}|&=\left|\frac{\sin(\Delta \omega D/2)}{\Delta \omega D/2}\right|.
\end{align}

\noindent This is a sinc function, where the most significant change of is reflected near $ \Delta\omega D/2=\pi $, corresponding to the resolution cell $ \Delta\omega=\Omega $.
Since our study indicates that $ \Delta\omega $ primarily influences the \ac{crb} in the form of $ Q(\Delta\omega) $, the significant variation of $ Q(\Delta\omega) $ near $ \Omega $ indirectly results in the rapid fluctuation of $ {\rm CRB}(\Delta\omega) $ around $ \Omega $.



\subsection{Verification of CRB turning point's low sensitivity to SNR}
\label{subsec:snr_sensitive}

We show the turning points of $ {\rm CRB}_{\rm FC} $ and $ {\rm CRB}_{\rm PC} $ in different \ac{snr}s to explain that the proposed resolution criterion is not sensitive to noise.

Consider the same simulation setting as Subsection \ref{subsec:example} except $ L=3 $.
We plot the $ {\rm CRB}_{\rm FC} $ and $ {\rm CRB}_{\rm PC} $ curves for \ac{snr} being $ 10 $, $ 20 $, and $ 30 $ dB, shown as Fig.~\ref{fig:different_snr}.
From Fig.~\ref{fig:different_snr}, we find that the curves of $ {\rm CRB}_{\rm PC} $ in different \ac{snr}s have the same shapes, as well as the turning points. 
How the noise affects the \ac{crb} reflects on the absolute values of \ac{crb}, instead of the relative relationship between \ac{crb} and $ \Delta\omega $.
This verifies that the proposed resolution criterion using the \ac{crb} turning point is not sensitive to \ac{snr} in Subsection~\ref{subsec:novel_usage_CRB}.


\begin{figure}[!htbp]
\centering
\includegraphics[width=3.2in]{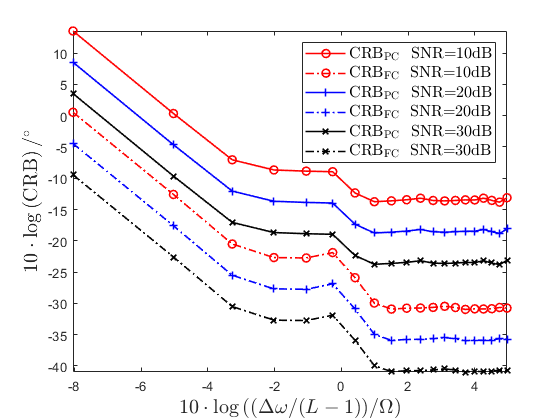}
\caption{$ {\rm CRB}_{\rm FC} $ and $ {\rm CRB}_{\rm PC} $ w.r.t. $ (\Delta\omega/(L-1))/\Omega $ for different \ac{snr}s when $ L=3 $.}
\label{fig:different_snr}
\end{figure}

\subsection{Verification of the approximation in Lemma \ref{lemma:PC}}
\label{subsec:approximation_all}

Since $\bm{M}\bm{G}^{-1}\bm{M}^T$ is the main component of $ {\rm CRB}_{\rm PC} $ different from $ {\rm CRB}_{\rm FC} $, we show its approximation errors in the proof of Lemma \ref{lemma:PC}. 
Particularly, we compare the approximate $ \bm{M}\bm{G}^{-1}\bm{M}^T $ with the true counterpart.

Consider the same simulation setting as Subsection \ref{subsec:example} except $ L=2 $.
We take the (1,1) entries of the true and approximate $ \bm{M}\bm{G}^{-1}\bm{M}^T $ as an example, and the comparison is shown in Fig.~\ref{fig:approx_all}.
From Fig.~\ref{fig:approx_all}, we find that the approximation errors are large when $ \Delta\omega<\Omega $, and become small when $ \Delta\omega\geqslant\Omega $. 
This is because in the proof of Lemma \ref{lemma:PC}, we use $ |Q(\Delta\omega)|\approx 0 $ in Lemma \ref{lemma:Q} to simplify the proof.
This conclusion is feasible for $ \Delta\omega\geqslant\Omega $, but not for $ \Delta\omega<\Omega $.
When $ \Delta\omega\geqslant\Omega $, the approximation results are close to the true values, yielding the feasibility of the approximation.


\begin{figure}[!htbp]
\centering
\includegraphics[width=3.2in]{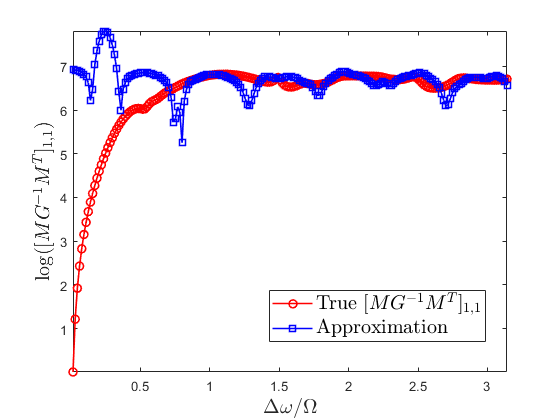}
\caption{The approximation performance of $ [\bm{M}\bm{G}^{-1}\bm{M}^T]_{1,1} $.}
\label{fig:approx_all}
\end{figure}

\subsection{Similar angular resolution between fully and partly calibrated arrays}
\label{subsec:algorithm}

We show that the angular resolution of fully and partly calibrated arrays is similar, which is achieved by comparing the upper resolution limit of the corresponding direction-finding algorithms.
Particularly, we provide the estimation accuracy of these algorithms varying with the source separation, and regard the turning point where the estimation accuracy initially tends to plateau out as the resolution limit.
For fair comparison, the directions are estimated using \ac{music} \cite{schmidt1986multiple} for fully calibrated arrays, and root-RARE \cite{RARE02}, spectral-RARE \cite{exRARE04}, and ESPRIT-GP \cite{6621815} for partly calibrated arrays, since these algorithms are all based on subspace separation techniques with the difference being whether errors exist. 

Consider the same simulation settings as Subsection~\ref{subsec:example} except $ L=2 $, the complex coefficients $ \bm{s} $ being standard Gaussian variables and the number of snapshots being $ T=50 $.
We use the \ac{rmse} of $ \Delta\omega $ to indicate the estimation performance of directions. 
We carry out $ T_m=300 $ Monte Carlo trials and denote the \ac{rmse} of $ \bm{\omega} $ by
\begin{equation}
\label{equ:rmse}
{\rm RMSE}(\bm{\omega})=\sqrt{\frac{1}{T_m}\sum_{t=1}^{T_m} \left\Vert \hat{\bm{\omega}}_t-\bm{\omega}^* \right\Vert_2^2},
\end{equation}
where $ \hat{\bm{\omega}}_t $ is the estimate in the $ t $-th trial and $ \bm{\omega}^* $ is the true value.
The estimation results of \ac{music} for fully calibrated arrays and root-RARE, spectral-RARE, and ESPRIT-GP for partly calibrated arrays w.r.t. $ \Delta\omega $ are shown in Fig.~\ref{fig:MUSIC}.

\begin{figure}[!htbp]
\centering
\subfigure[Partly calibrated case.]{
\label{subfigure:musica}
\includegraphics[width=1.6in]{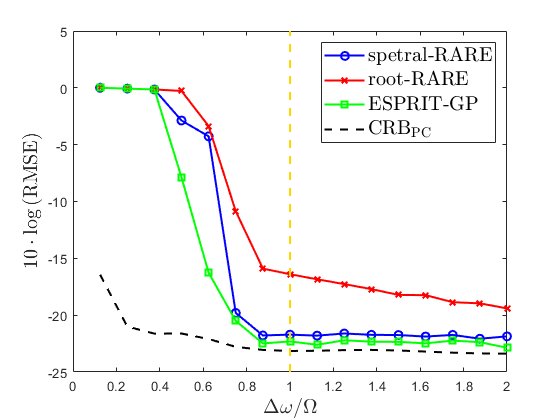}}
\subfigure[Fully calibrated case.]{
\label{subfigure:musicb}
\includegraphics[width=1.6in]{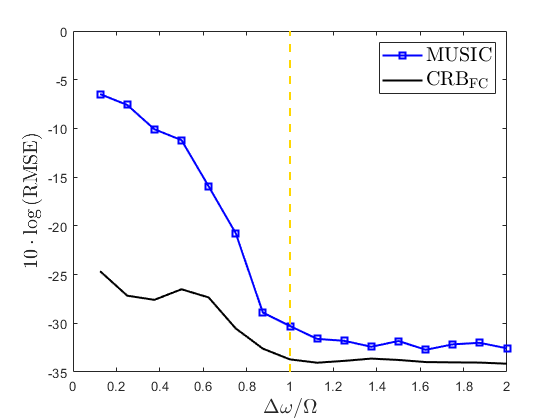}}
\caption{The \ac{rmse} of $ \bm{\omega} $ w.r.t. $ \Delta\omega $ in fully and partly calibrated arrays.}
\label{fig:MUSIC}
\end{figure}

From Fig.~\ref{subfigure:musica}, we find that when $ \Delta\omega\leqslant\Omega $, the \ac{rmse}s of root-RARE, spectral-RARE, and ESPRIT-GP decline as $ \Delta\omega $ increases; when $ \Delta\omega>\Omega $, the \ac{rmse}s tend to be stable.
Similar phenomenon is found in the fully calibrated case in Fig.~\ref{subfigure:musicb}, yielding that the resolution limit of these algorithms is close to $ \Omega $.
This also heuristically indicates that the resolution limit of fully and partly calibrated arrays is similar, both close to $ \Omega $, corresponding to our main conclusion.
We note that the turning points of the \ac{rmse}s w.r.t. $ \Delta\omega $ in Fig.~\ref{fig:MUSIC}  are not exactly located at $ \Omega $ since $ \Omega $ is an empirical bound and the subspace based algorithms have super-resolution ability.

Then, we construct the dependence of the resolution probability for the RARE/ESPRIT methods in partly calibrated arrays and MUSIC method in fully calibrated arrays.
Particularly, we define the resolution probability $ P_t $ as the probability that the estimation error is less than a threshold, given by
\begin{align}
\label{equ:prothre}
P_t=\frac{T_p}{T_m},
\end{align}
where $ T_p $ is the number of trials in which the estimation error is less than $ -13 $dB for partly calibrated case and $ -30 $dB for fully calibrated case, respectively, and the thresholds are empirically chosen based on the corresponding \ac{crb}s.
The resolution probability of RARE/ESPRIT and MUSIC is shown as Fig.~\ref{fig:probability}, which also verifies that the angular resolution between fully and partly calibrated arrays is similar.

\begin{figure}[!htbp]
\centering
\subfigure[Partly calibrated case.]{
\label{subfigure:musicaa}
\includegraphics[width=1.6in]{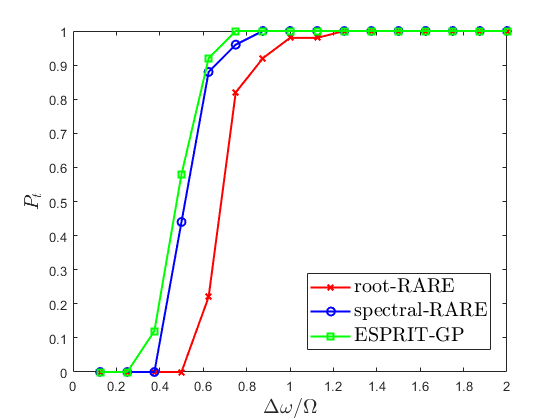}}
\subfigure[Fully calibrated case.]{
\label{subfigure:musicbb}
\includegraphics[width=1.6in]{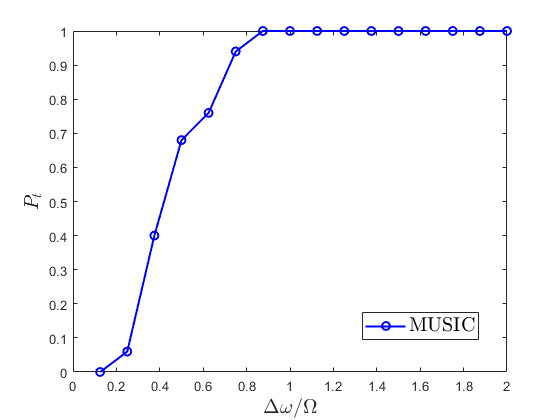}}
\caption{The resolution probability in fully and partly calibrated arrays.}
\label{fig:probability}
\end{figure}

\section{Conclusion}
\label{sec:conclusion}

In this paper, we theoretically explain that partly calibrated arrays achieve similar angular resolution as fully calibrated arrays, both inversely proportional to the whole array aperture.
The analysis is based on a characteristics of \ac{crb} that when the source separation $ \Delta\omega $ increases, the \ac{crb} w.r.t $ \Delta\omega $ first declines rapidly, then plateaus out, and the turning point is close to the angular resolution limit.
Hence, we transform the angular resolution analysis into comparing the turning points of the \ac{crb}s of fully and partly calibrated arrays, where the key technique lies in the partial derivative of \ac{crb} w.r.t. $ \Delta\omega $.
To this end, we introduce some intermediate variables to simplify the analysis and theoretically explain the declining and plateau phases of the \ac{crb}.
This work provides an important theoretical guarantee for the high-resolution performance of distributed arrays in mobile platforms.
We believe that the proposed method to analyze the impact of position errors on resolution in distributed arrays using \ac{crb} turning points can be applied to other types of errors, such as gain and phase errors. 
This is theoretically feasible, but its specific implementation requires further research.




\appendices

\section{Calculation of CRB}
\label{app:cal_CRB}

The \ac{crb}s of spatial frequencies $\bm{\omega}$ using fully and partly calibrated arrays, denoted by $ {\rm CRB}_{\rm FC}(\bm{\omega}) $ and $ {\rm CRB}_{\rm PC}(\bm{\omega}) $, respectively, are given by \cite{RARE02},
\begin{align}
\label{equ:CRB_FC}
{\rm CRB}_{\rm FC}(\bm{\omega})&=\frac{\sigma^2}{2}\bm{F}^{-1}, \\
\label{equ:CRB_PC}
{\rm CRB}_{\rm PC}(\bm{\omega})&=\frac{\sigma^2}{2}(\bm{F}-\bm{M}\bm{G}^{-1}\bm{M}^T)^{-1},
\end{align}
where
\begin{align}
\label{equ:F}
\bm{F}&=\sum_{i=1}^T{\rm Re}\left\{\bm{D}_i^H\bm{\Pi}_{\bm{A}}^{\perp}\bm{D}_i\right\},
\end{align}
\begin{align}
\label{equ:M}
\bm{M}&=\sum_{i=1}^T{\rm Re}\left\{\bm{D}_i^H\bm{\Pi}_{\bm{A}}^{\perp}\bm{H}_i\right\}, 
\end{align}
\begin{align}
\label{equ:G}
\bm{G}&=\sum_{i=1}^T{\rm Re}\left\{\bm{H}_i^H\bm{\Pi}_{\bm{A}}^{\perp}\bm{H}_i\right\},
\end{align}
\begin{align}
\label{equ:PiA}
\bm{\Pi}_{\bm{A}}^{\perp}&=\bm{I}-\bm{A}(\bm{A}^H\bm{A})^{-1}\bm{A}^H, 
\end{align}
\begin{align}
\label{equ:D}
\bm{D}_i&=\left[\frac{\partial \bm{a}(\omega_1)}{\partial \omega_1}s_1(t_i),\dots,\frac{\partial \bm{a}(\omega_L)}{\partial \omega_L}s_L(t_i)\right], 
\end{align}
\begin{align}
\frac{\partial \bm{a}(\omega)}{\partial \omega}&=j\cdot\left[\varphi_1 e^{j\omega\varphi_1},\dots,\varphi_N e^{j\omega\varphi_N}\right]^T, 
\end{align}
\begin{align}
\label{equ:H}
\bm{H}_i&=\begin{bmatrix}
\bm{0} & & \bm{0} \\
\widetilde{\bm{B}}_{\bm{\xi},i}^2 & & \bm{0} \\
& \ddots & \\
\bm{0} & & \widetilde{\bm{B}}_{\bm{\xi},i}^K
\end{bmatrix},
\end{align}
\begin{align}
\widetilde{\bm{B}}_{\bm{\xi},i}^k&=j\cdot\sum_{l=1}^L\ s_l(t_i)\cdot\omega_l\cdot\bm{a}_k(\omega_l).
\label{equ:Bxi}
\end{align}

We consider the single-snapshot case and omit $ i $ since the multiple-snapshot cases ($ T>1 $) is a direct extension to those in the single-snapshot cases ($ T=1 $). 

\section{Proof of Proposition \ref{prop:prop2}}
\label{app:proof_of_prop2}

The proof of Proposition \ref{prop:prop2} is an extension of Lemma \ref{prop:prop1} \cite{Lee92} from fully calibrated arrays to partly calibrated arrays.
The main difference lies in the small quantity approximation of the matrix $ \bm{M}\bm{G}^{-1}\bm{M}^T $ in \eqref{equ:CRB_PC}. 

Particularly, denote $ \Delta\omega_l=\omega_l-\omega_1 $ for $ l=1,\dots,L $ and $ \Delta\omega\equiv\Delta\omega_L=\omega_L-\omega_1 $.
When $ \Delta\omega\ll\Omega $, we carry out Taylor expansion of $ \bm{a}_k(\omega) $ in \eqref{equ:Bxi} on $ \omega=\omega_1 $ as
\begin{align}
\label{equ:Taylor}
\bm{a}_k(\omega_l)&=\bm{a}_k(\omega_1)+\bm{a}_k^{(1)}(\omega_1)\Delta\omega_l+\bm{a}_k^{(2)}(\omega_1)\frac{\Delta\omega_l^2}{2!}+\cdots, \nonumber \\
&=\bm{a}_k(\omega_1)+\tilde{\bm{a}}_k^{(1)}(\omega_1)\Delta\omega+\tilde{\bm{a}}_k^{(2)}(\omega_1)\frac{\Delta\omega^2}{2!}+\cdots,
\end{align}
where $ \tilde{\bm{a}}_k^{(p)}(\omega_0)\equiv\bm{a}_k^{(p)}(\omega_0)\cdot (\Delta\omega_l/\Delta\omega)^p $ and $ (\cdot)^{(p)} $ denotes the $ p $-order derivative of $ \cdot $ for $ p=1,2,\dots $.
By substituting \eqref{equ:Taylor} to \eqref{equ:Bxi}, we have 
\begin{align}
\label{equ:Bxiomega}
\widetilde{\bm{B}}_{\bm{\xi}}^k=\bm{c}_{\xi}^k+O(\Delta\omega),
\end{align}
where $ \bm{c}_{\xi}^k\equiv j\cdot\sum_{l=1}^L\ s_l\cdot\omega_1\cdot\bm{a}_k(\omega_1) $. 

From the proof of Lemma~\ref{prop:prop1} \cite{Lee92}, 
$ \bm{D} $ in \eqref{equ:D} and $ \bm{\Pi}_{\bm{A}}^{\perp} $ in \eqref{equ:PiA} w.r.t. $ \Delta\omega $ are expressed as follows:
\begin{align}
\label{equ:insideprop1_D}
\bm{D}&=(\Delta\omega)^{L-1}\bm{C}_{D}+O((\Delta\omega)^L), \\
\label{equ:insideprop1_Pi}
\bm{\Pi}_{\bm{A}}^{\perp}&=\bm{C}_A+O(\Delta\omega),
\end{align}
where $ \bm{C}_{D} $ and $ \bm{C}_A $ are constants w.r.t. $ \Delta\omega $. 
Since $ \bm{\Pi}_{\bm{A}}^{\perp} $ in \eqref{equ:PiA} is a symmetric projection matrix satisfying $ \bm{\Pi}_{\bm{A}}^{\perp}=(\bm{\Pi}_{\bm{A}}^{\perp})^H=\bm{\Pi}_{\bm{A}}^{\perp}\bm{\Pi}_{\bm{A}}^{\perp} $, we rewrite $ \bm{M} $ in \eqref{equ:M} and $ \bm{G} $ in \eqref{equ:G} respectively as
\begin{align}
\label{equ:M_omega}
\bm{M}&=\bm{D}^H\bm{\Pi}_{\bm{A}}^{\perp}\bm{H}=(\bm{\Pi}_{\bm{A}}^{\perp}\bm{D})^H\bm{\Pi}_{\bm{A}}^{\perp}\bm{H}, \\
\label{equ:G_omega}
\bm{G}&=\bm{H}^H\bm{\Pi}_{\bm{A}}^{\perp}\bm{H}=(\bm{\Pi}_{\bm{A}}^{\perp}\bm{H})^H\bm{\Pi}_{\bm{A}}^{\perp}\bm{H}.
\end{align}
By substituting \eqref{equ:H}, \eqref{equ:insideprop1_Pi} and \eqref{equ:Bxiomega} to $ \bm{\Pi}_{\bm{A}}^{\perp}\bm{H} $, we have
\begin{align}
\label{equ:PiAH}
\bm{\Pi}_{\bm{A}}^{\perp}\bm{H}=\bm{C}_H+O(\Delta\omega),
\end{align}
and $ \bm{C}_H=[\bm{C}_A^2\bm{c}_{\xi}^2,\dots,\bm{C}_A^K\bm{c}_{\xi}^K] $ is a constant w.r.t. $ \Delta\omega $, where $ \bm{C}_A=[\bm{C}_A^1,\dots,\bm{C}_A^K] $ and $ \bm{C}_A^k $ denotes the $ n\in\mathcal{N}_k $ columns of $ \bm{C}_A $ for $ k=1,\dots,K $.

If $ \bm{C}_H\ne\bm{0} $ (to be proved later), we substitute \eqref{equ:insideprop1_D}, \eqref{equ:insideprop1_Pi} and \eqref{equ:PiAH} to \eqref{equ:M_omega} and \eqref{equ:G_omega}, and have 
\begin{align}
\label{equ:Msmall}
\bm{M}&=(\Delta\omega)^{L-1}{\rm Re}\{\bm{C}^H_{D}\bm{C}_A^H\bm{C}_H\}+O((\Delta\omega)^L), \\
\label{equ:Gsmall}
\bm{G}&={\rm Re}\left\{(\bm{C}_H+O(\Delta\omega))^H(\bm{C}_H+O(\Delta\omega))\right\},
\end{align}
yielding that
\begin{align}
\label{equ:MG-1MT}
\bm{M}\bm{G}^{-1}\bm{M}^T=(\Delta\omega)^{2(L-1)}\bm{C}_U+O((\Delta\omega)^{2(L-1)+1}),
\end{align}
where $ \bm{C}_U $ is a constant w.r.t. $ \Delta\omega $.
Based on \eqref{equ:MG-1MT} and Lemma \ref{prop:prop1}, we have \eqref{equ:insideprop2} with
\begin{align}
\label{equ:CG}
\bm{C}_G=\frac{\sigma^2}{2}\left(\bm{C}_U-{\rm Re}\{\bm{C}_{D}^H\bm{C}_A\bm{C}_{D}\}\right)^{-1},
\end{align}
completing the proof.

Here we prove that any column of $ \bm{C}_H $ is not $ \bm{0} $, hence $ \bm{C}_H\ne\bm{0} $:
Consider if there is one column of $ \bm{C}_H $ equal to $ \bm{0} $.
We assume that the 1-st column of $ \bm{C}_H $ is $ \bm{0} $ without loss of generality, given by $ \bm{C}_A^2\bm{c}_{\xi}^2=\bm{0} $. In this case, we construct the following vector,
\begin{align}
\label{equ:2Lvectors}
\bar{\bm{a}}=\left[\bm{0}^T,(\bm{c}_{\xi}^2)^T,\bm{0}^T\right]^T\in\mathbb{C}^{N\times1},
\end{align}
such that $ \bm{C}_A\bar{\bm{a}}=\bm{0} $, where the the first $ \bm{0}\in\mathbb{R}^{|\mathcal{N}_1|\times1} $, the second $ \bm{0}\in\mathbb{R}^{(N-|\mathcal{N}_1|-|\mathcal{N}_2|)\times1} $, and $ \bm{C}_{A} $ is expressed as \cite{RARE02}
\begin{align}
\label{equ:CA_exp}
\bm{C}_A=\bm{I}-\dot{\bm{A}}(\dot{\bm{A}}^H\dot{\bm{A}})^{-1}\dot{\bm{A}}^H,
\end{align}
with 
\begin{align}
\dot{\bm{A}}&=\left[\bm{a}(\omega_1),\bm{a}^{(1)}(\omega_1),\dots,\bm{a}^{(L-1)}(\omega_1)\right], \\
\bm{a}^{(l)}(\omega_1)&=\left[\frac{\partial^l}{\partial\omega^l}\bm{a}(\omega)\right]_{\omega=\omega_1}.
\end{align}
Since $ \bm{C}_A\bar{\bm{a}}=\bm{0} $ and $ \bm{C}_A $ in \eqref{equ:CA_exp} is a projection matrix, $ \bar{\bm{a}} $ should be in the column space of $ \dot{\bm{A}} $. 
However, we then explain that $ \bar{\bm{a}} $ could not be in the column space of $ \dot{\bm{A}} $, i.e., there is no $ \bm{y}\in\mathbb{C}^{L\times1} $ such that $ \bar{\bm{a}}=\dot{\bm{A}}\bm{y} $, yielding a contradiction.
Particularly, define $ \dot{\bm{A}}=[\dot{\bm{A}}_1^T,\dots,\dot{\bm{A}}_K^T]^T $ and $ \dot{\bm{A}}_k $ denotes the $ n\in\mathcal{N}_k $ rows of $ \dot{\bm{A}} $ for $ k=1,\dots,K $.
We find that $ \dot{\bm{A}}_1\in\mathbb{C}^{|\mathcal{N}_1|\times L}, |\mathcal{N}_1|>L $ has full column rank, given by
\begin{align}
\dot{\bm{A}}_1=\bm{D}_A\begin{bmatrix}
1 & j\varphi_1 & \dots & (j\varphi_1)^{L-1} \\
& & \vdots & \\
1 & j\varphi_{|\mathcal{N}_1|} & \dots & (j\varphi_{|\mathcal{N}_1|})^{L-1} 
\end{bmatrix},
\end{align}
where $ \bm{D}_A={\rm diag}([e^{j\omega_0\varphi_1},\dots,e^{j\omega_0\varphi_{|\mathcal{N}_1|}}]^T) $.
Therefore, there is no nonzero $ \bm{y}\in\mathbb{C}^{L\times1} $ such that $ [\dot{\bm{A}}]_1\bm{y}=\bm{0} $, yielding $ \bar{\bm{a}}\ne\dot{\bm{A}}\bm{y} $ for any nonzero $ \bm{y}\in\mathbb{C}^{L\times1} $.
This contradiction implies any column of $ \bm{C}_H $ is not $ \bm{0} $ and hence $ \bm{C}_H\ne\bm{0} $.

\section{Proof of Lemma~\ref{lemma:Q}}
\label{app:prop_4_1}

Lemma~\ref{lemma:Q} in Subsection \ref{subsec:the_proof1} is a rough representation, and we introduce its rigorous expression here, given by Lemma~\ref{lemma:Q}.1 and Lemma~\ref{lemma:Q}.2.
For $ Q_i, i\in\mathbb{N} $, we take $ Q_0 $ as an example and the other cases are its direct extension with minor modifications.

For $ |Q_0(\Delta\omega)|\approx 0 $ in Lemma~\ref{lemma:Q}, the rigorous expression is as follows:

\newtheorem{lem}{Lemma 2.}

\begin{lem}
\label{lemma:lemmaEQ}
Under assumptions A1-A3, for any $ t>0 $, define $ t_1,t_2 $ as
\begin{align}
\label{equ:t1}
t_1&\equiv \max\left\{\left|t+\frac{\sin\Delta\omega D}{\Delta\omega D}\right|,\left|-t+\frac{\sin\Delta\omega D}{\Delta\omega D}\right|\right\}, \\
\label{equ:t2}
t_2&\equiv\left|t+\frac{1-\cos\Delta\omega D}{\Delta\omega D}\right|.
\end{align}
When $ \Delta\omega\geqslant \Omega $, we have
\begin{align}
P\left(|Q_0|\geqslant \sqrt{t_1^2+t_2^2}\right)\leqslant 4e^{-\frac{Nt^2}{2}}.
\end{align}
\end{lem}
\begin{proof}


First, we rewrite $ Q_0 $ in \eqref{equ:threeQ} as
\begin{align}
\label{equ:Q_0_2}
Q_0=\frac{\sum_{n=1}^N\cos\Delta\omega\varphi_n+j\sum_{n=1}^N\sin\Delta\omega\varphi_n}{N}.
\end{align}
When $ \varphi_n\in\mathcal{U}[0,D] $ in the A2 assumption, we have
\begin{align}
\label{equ:cos_average} 
\mathbb{E}\left(\frac{1}{N}\sum_{n=1}^N\cos\Delta\omega\varphi_n\right)&=\frac{\sin\Delta\omega D}{\Delta\omega D}, \\
\label{equ:sin_average}
\mathbb{E}\left(\frac{1}{N}\sum_{n=1}^N\sin\Delta\omega\varphi_n\right)&=\frac{1-\cos\Delta\omega D}{\Delta\omega D}.
\end{align}
Through the Hoeffding inequality \cite{boucheron2003concentration} and some straightforward derivation, we have 
\begin{align}
&P_1\equiv P\left(\left|\frac{1}{N}\sum_{n=1}^N\cos\Delta\omega\varphi_n\right|\geqslant t_1 \right)\leqslant 2e^{-\frac{Nt^2}{2}}, \\
&P_2\equiv P\left(\left|\frac{1}{N}\sum_{n=1}^N\sin\Delta\omega\varphi_n\right|\geqslant t_2 \right)\leqslant 2e^{-\frac{Nt^2}{2}},
\end{align}
where $ t_1,t_2 $ are defined in \eqref{equ:t1} and \eqref{equ:t2}, respectively.
For $ Q_0 $ in \eqref{equ:Q_0_2}, we have
\begin{align}
\label{equ:hoff_end}
P\left(|Q_0|\geqslant \sqrt{t_1^2+t_2^2}\right) \leqslant P_1+P_2 \leqslant 4e^{-\frac{Nt^2}{2}},
\end{align}
completing the proof.
\end{proof}

For $ \left|\partial Q_0(\Delta\omega)/\partial \Delta\omega\right|\approx 0 $ in Lemma~\ref{lemma:Q}, the rigorous expression is as follows:
\begin{lem}
\label{lemma:lemmaEQ2}
Under assumptions A1-A3, for any $ t>0 $, when $ \Delta\omega\geqslant\Omega $, we have
\begin{align}
P\left(\left|\frac{\partial Q_0(\Delta\omega)}{\partial \Delta\omega}\right|\geqslant \frac{\sum_{n=1}^N \varphi_n}{N}\sqrt{\bar{t}_1^2+\bar{t}_2^2}\right)\leqslant 4e^{-\frac{Nt^2}{2}},
\end{align}
where $ \bar{t}_1,\bar{t}_2 $ are calculated for $ Q_1 $ similar with $ t_1,t_2 $ for $ Q_0 $ in Lemma 2.~\ref{lemma:lemmaEQ}.
\end{lem}
\begin{proof}
\label{proof:lemma22}
The partial derivative of $ Q_0(\Delta\omega) $ w.r.t. $ \Delta\omega $ is 
\begin{align}
\label{equ:Q_partial}
\left|\frac{\partial Q_0(\Delta\omega)}{\partial \Delta\omega}\right|&=\left| \frac{\sum_{n=1}^N \varphi_n e^{j\Delta\omega\varphi_n}}{N} \right| \nonumber \\
&=\frac{\sum_{n=1}^N \varphi_n}{N}\cdot \left|Q_{1}\right|,
\end{align}
completing the proof.
\end{proof}

\section{Proof of Lemma \ref{lemma:PC}}
\label{app:proof_of_prop4}

$ {\rm CRB}_{\rm PC} $ in \eqref{equ:CRB_PC} can be divided into $ \bm{F} $ and $ \bm{M}\bm{G}^{-1}\bm{M}^T $.
The former is transformed into a function of $ 
Q_i $ in Lemma~\ref{lemma:FC}, and we consider the later one here.

To simply the derivation, we use the approximation $ |Q_0|\approx 0 $ in Lemma~\ref{lemma:Q}, and then substituting \eqref{equ:piAQ} to \eqref{equ:M} and \eqref{equ:G}, 
\begin{align}
\label{equ:M_approx}
\bm{M}&\approx\widetilde{\bm{M}}\equiv{\rm Re}\left\{\bm{D}^H\bm{H}-\frac{\bm{D}^H\bm{A}\bm{A}^H\bm{H}}{N}\right\}, \\
\label{equ:G_approx}
\bm{G}&\approx\widetilde{\bm{G}}\equiv{\rm Re}\left\{\bm{H}^H\bm{H}-\frac{\bm{H}^H\bm{A}\bm{A}^H\bm{H}}{N}\right\},
\end{align}
where $ \bm{H}^H\bm{H} $ is a diagonal matrix denoted by $ \bm{H}^H\bm{H}={\rm diag}(\bm{h}) $. 

Based on $ |Q_1|\approx 0 $ in Lemma \ref{lemma:Q} and $  Q_1^k\approx Q_0^k $ in the A3 assumption (let $ f(x)=g(x)=x $), by substituting \eqref{equ:D} and \eqref{equ:H} to \eqref{equ:M_approx}, we approximate $ \bm{M} $ as
\begin{align}
\label{equ:Mrecast}
[\bm{M}]_{k-1}&\approx[\widehat{\bm{M}}]_{k-1} \nonumber \\
&\equiv\gamma_M^k\cdot \omega_0\cdot
\begin{bmatrix}
|s_1|^2+|s_1s_2|s_0^k\mu \\
|s_1s_2|s_0^k +|s_2|^2\mu
\end{bmatrix},
\end{align}
where $ \gamma_M^k=|\mathcal{N}_k|\cdot(\bar{\varphi}_k-\bar{\varphi}) $, $ s_0^k=\frac{{\rm Re}\{s_1^Hs_2 Q_0^k\}}{|s_1s_2|} $, $ \mu=1+\frac{\Delta\omega}{\omega_0} $ and $ \bar{\varphi}=\frac{\sum_{n=1}^N\varphi_n}{N} $ for $ k=2,\dots,K $.

Then, we approximate $ \bm{G}^{-1} $ by $ \widehat{\bm{G}}^{-1} $, which is shown in the following lemma:
\begin{lemma}
\label{lemma:lemmaBB}
Based on Assumption A2, we have 
\begin{align}
\label{equ:G-1}
\bm{G}^{-1}&\approx\widehat{\bm{G}}^{-1}\equiv\widehat{\bm{G}}_1+\widehat{\bm{G}}_2 \nonumber \\
&=\begin{bmatrix}
\frac{1}{B_2} & & 0 \\
 & \ddots & \\
0 & & \frac{1}{B_k}
\end{bmatrix}+ {\rm Re}\left\{
\begin{bmatrix}
\frac{\bm{b}_2^H}{B_2} \\
\vdots \\
\frac{\bm{b}_K^H}{B_K}
\end{bmatrix}
\begin{bmatrix}
\frac{\bm{b}_2}{B_2} & \cdots & \frac{\bm{b}_K}{B_K}
\end{bmatrix}\right\},
\end{align}
where $ B_k=[\bm{h}]_{k-1}>0 $ and $ \bm{b}_k=\frac{1}{\sqrt{N}}[\bm{A}^H\bm{H}]_{k-1} $.
\end{lemma}
\begin{proof}
See Appendix \ref{app:BB}.
\end{proof}

Now, we approximate $ \bm{M}\bm{G}^{-1}\bm{M}^T $ by 
\begin{align}
\label{equ:MGMT_approx}
\bm{M}\bm{G}^{-1}\bm{M}^T\approx\widehat{\bm{M}}\widehat{\bm{G}}^{-1}\widehat{\bm{M}}^T=\widehat{\bm{M}}(\widehat{\bm{G}}_1+\widehat{\bm{G}}_2)\widehat{\bm{M}}^T.
\end{align}
Define $ \mathcal{G}_1=\widehat{\bm{M}}\widehat{\bm{G}}_1\widehat{\bm{M}}^T $ and $ \mathcal{G}_2=\widehat{\bm{M}}\widehat{\bm{G}}_2\widehat{\bm{M}}^T $.
In the sequel, we discuss how $ \mathcal{G}_1 $ and $ \mathcal{G}_2 $ are recast as functions of $ Q_i $, respectively.
Based on the A1 assumption, $ L=2 $ and $ \mathcal{G}_1,\mathcal{G}_2\in\mathbb{R}^{2\times 2} $. 
We take the $ (1,1) $-th entry as an example, and the other entries can be directly derived.
This is because the rows of $ \bm{M} $ have similar form, i.e., 
\begin{align}
\label{equ:extension_M}
[\bm{M}]_{2,k-1}(s_1,s_2,\mu) = \mu\cdot[\bm{M}]_{1,k-1}(s_2^H,s_1^H,1/\mu),
\end{align}
for $ k=2,\dots,K $.

For $ [\mathcal{G}_1]_{1,1} $, substituting \eqref{equ:Mrecast} and \eqref{equ:G-1} to $ \mathcal{G}_1 $ yields
\begin{align}
\label{equ:G11_recast1}
[\mathcal{G}_1]_{1,1}=|s_1|^2\sum_{k=2}^K\frac{(\gamma_M^k)^2}{|\mathcal{N}_k|} \left(1-\frac{\tilde{s}(1-(s_0^k)^2)\mu^2}{\tilde{s}^{-1}+\tilde{s}\mu^2+2s_0^k\mu}\right),
\end{align}
where $ \tilde{s}=\left|s_2/s_1\right|>0 $. 
Due to $ |Q_0^k|<1 $ in \eqref{equ:Q<1}, we have $ |s_0^k|<1 $, yielding that
\begin{align}
\label{equ:s0s1s2}
\tilde{s}^{-1}+\tilde{s}\mu^2+2s_0^k\mu> 0.
\end{align}
Consider the Taylor expansion of $ \frac{1}{\tilde{s}^{-1}+\tilde{s}\mu^2+2s_0^k\mu} $ at $ 2s_0^k\mu=0 $, given by
\begin{align}
\label{equ:smallonS0}
\frac{1}{\tilde{s}^{-1}+\tilde{s}\mu^2+2s_0^k\mu}&= \sum_{p=0}^{\infty}\frac{(-2s_0^k\mu)^p}{(\tilde{s}^{-1}+\tilde{s}\mu^2)^{p+1}}.
\end{align}
By substituting \eqref{equ:smallonS0} to \eqref{equ:G11_recast1}, we have
\begin{align}
\label{equ:approx_G11_Taylor}
[\mathcal{G}_1]_{1,1}=|s_1|^2\sum_{k=2}^K\frac{(\gamma_M^k)^2}{|\mathcal{N}_k|}\left(1-\tilde{s}(1-(s_0^k)^2)\sum_{p=0}^{\infty}u_p(s_0^k)^p\right),
\end{align}
where $ u_p $ is denoted by
\begin{align}
\label{equ:up}
u_p(\mu)=\frac{(-2)^p\mu^{p+2}}{(\tilde{s}^{-1}+\tilde{s}\mu^2)^{p+1}}.
\end{align}

The effect of $ \Delta\omega $ on $ [\mathcal{G}_1]_{1,1} $ in \eqref{equ:approx_G11_Taylor} is embodied in both $ s_0^k $ and $ u_p $.
In the sequel, we explain that $ u_p $ is almost constant w.r.t. $ \Delta\omega $, which is detailed in the following lemma. 
\begin{lemma}
\label{lemma:lemma_gradient_up}
The gradient of $ u_p $ w.r.t. $ \mu $, $ u_p'=\frac{{\rm d}u_p}{{\rm d}\mu} $, satisfies
\begin{align}
\label{equ:pro1}
&1.\ |u_p'|\leqslant 1,\ p=0,1,2,3. \\
\label{equ:pro2}
&2.\ \underset{|\mu|\rightarrow\infty}{\rm Lim}\ |u_p'|=0,\ p=0,1,2,\dots.
\end{align}
\end{lemma}
\begin{proof}
See Appendix \ref{app:up}.
\end{proof}

Based on Lemma \ref{lemma:lemma_gradient_up}, when $ \Delta\omega\geqslant\Omega $, we approximate $ u_p $ by the constant $ {u}_p(\bar{\mu}) $, where $ \bar{\mu} \equiv\mu(\Omega)=1+\frac{\Omega}{\omega_0} $.
Since $ |s_0^k|<1 $, we ignore the high-order terms of $ (s_0^k)^p $ and approximate the infinite summation in \eqref{equ:approx_G11_Taylor} by its $ p=0,1,2,3 $ terms as
\begin{align}
\label{equ:G_11_recast_g}
[\mathcal{G}_1]_{1,1}&\approx|s_1|^2\sum_{k=2}^K\frac{(\gamma_M^k)^2}{|\mathcal{N}_k|}\left(1-\tilde{s}(1-(s_0^k)^2)\sum_{p=0}^{3}\bar{u}_p(s_0^k)^p\right) \nonumber \\
&=\sum_{k=2}^K\frac{(\gamma_M^k)^2}{|\mathcal{N}_k|}\left(\sum_{p=0}^{5}v_p(s_0^k)^p\right),
\end{align}
where $ v_p/|s_1|^2 $ for $ p=0,1,\dots,5 $ are given by $ 1-\tilde{s}\bar{u}_0 $, $-\tilde{s}\bar{u}_1$, $-\tilde{s}\bar{u}_2+\tilde{s}\bar{u}_0$, $-\tilde{s}\bar{u}_3+\tilde{s}\bar{u}_1$, $\tilde{s}\bar{u}_2$, $\tilde{s}\bar{u}_3 $, respectively.
Note that $ s_0^k $ can be recast as
\begin{align}
\label{equ:QokReIm}
s_0^k=\frac{1}{2|s_1s_2|}\left(s_1^Hs_2 Q_0^k+(s_1^Hs_2 Q_0^k)^H\right).
\end{align}
Substituting \eqref{equ:QokReIm} to \eqref{equ:G_11_recast_g} yields 
\begin{align}
\label{equ:G11final}
[\mathcal{G}_1]_{1,1}\approx\sum_{k=2}^K |\mathcal{N}_k| \cdot\tilde{f}(\bar{\varphi}_k)\cdot \tilde{g}(Q_0^k),
\end{align}
where $ \tilde{f}(x)=(x-\bar{\varphi})^2 $ and $ \tilde{g}(x) $ is a polynomial function w.r.t. $ x^{-p} $ and $ x^p $ for $ p=0,1,\dots,5 $, the specific form of which is omitted.
Based on \eqref{equ:average} in Assumption A3, 
we have 
\begin{align}
\label{equ:}
[\mathcal{G}_1]_{1,1}&\approx\sum_{k=2}^K \sum_{n\in\mathcal{N}_k} \tilde{f}(\varphi_n)\cdot \tilde{g}(e^{j\Delta\omega\varphi_n}) \nonumber \\
&\approx\sum_{n=1}^N \tilde{f}(\varphi_n)\cdot \tilde{g}(e^{j\Delta\omega\varphi_n}),
\end{align}
which can be directly recast as a polynomial function of $ Q_i(p\Delta\omega) $ for $ i=0,1,2 $ and $ p=-5,-4,\dots,5 $, given by
\begin{align}
\label{equ:G1_11_Q}
[\mathcal{G}_1]_{1,1}\approx\sum_{i=0}^2\sum_{p=-5}^5 c_p^i\cdot Q_i(p\Delta\omega),
\end{align}
where $ c_p^i $ is a constant unrelated to $ \Delta\omega $, completing the proof of the part of $ \mathcal{G}_1 $.

For $ [\mathcal{G}_2]_{1,1} $, substituting \eqref{equ:G-1} to $ \mathcal{G}_2 $ yields
\begin{align}
\label{equ:G2_11}
[\mathcal{G}_2]_{1,1}&=\sum_{k_1=2}^K\sum_{k_2=2}^K\frac{[\widehat{\bm{M}}]_{1,k_1-1}[\widehat{\bm{M}}]_{1,k_2-1}}{N[\bm{h}]_{k_1-1}[\bm{h}]_{k_2-1}} \nonumber \\
&\cdot\Big(\sum_{l=1,2}{\rm Re}\left\{[\bm{A}^H\bm{H}]_{l,k_1-1}^H[\bm{A}^H\bm{H}]_{l,k_2-1}\right\}\Big).
\end{align}
Due to $ {\rm Re}(x^Hy)={\rm Re}(x){\rm Re}(y)+{\rm Im}(x){\rm Im}(y) $ for any $ x,y\in\mathbb{C} $, \eqref{equ:G2_11} can be recast as
\begin{align}
\label{equ:mathcalG211approx}
[\mathcal{G}_2]_{1,1}&=\sum_{l=1,2}\left(\sum_{k=2}^K\frac{[\widehat{\bm{M}}]_{1,k-1}{\rm Re}\left\{[\bm{A}^H\bm{H}]_{l,k-1}\right\}}{\sqrt{N}[\bm{h}]_{k-1}}\right)^2 \nonumber \\
&+\sum_{l=1,2}\left(\sum_{k=2}^K\frac{[\widehat{\bm{M}}]_{1,k-1}{\rm Im}\left\{[\bm{A}^H\bm{H}]_{l,k-1}\right\}}{\sqrt{N}[\bm{h}]_{k-1}}\right)^2.
\end{align}
Similar as $ [\mathcal{G}_1]_{1,1} $ in \eqref{equ:approx_G11_Taylor}, we substitute \eqref{equ:H} and \eqref{equ:Mrecast} to \eqref{equ:mathcalG211approx} and have
\begin{align}
\label{equ:G2_11_recast}
[\mathcal{G}_2]_{1,1}=\frac{|s_1|^2}{4N}\sum_{r=1}^4\left(\sum_{K=2}^K\gamma_M^k\sum_{p=0}^{\infty}\frac{u_p\cdot(s_0^k)^p\cdot \bar{g}_r}{\mu^2}\right)^2,
\end{align}
where $ \bar{g}_r $ for $ r=1,\dots,4 $ are expressed as $ (s_1-s_1^H)\pm\mu(s_2Q_0^k-(s_2Q_0^k)^H) $ and $ (s_1(Q_0^k)^H+s_2\mu)\pm(s_1^HQ_0^k+s_2^H\mu) $.
Similar to the approximation on $ [\mathcal{G}_1]_{1,1} $, we apply the procedures from \eqref{equ:G_11_recast_g} to \eqref{equ:G1_11_Q} to $ [\mathcal{G}_2]_{1,1} $ in \eqref{equ:mathcalG211approx} and have
\begin{align}
\label{equ:G2_11_Q}
[\mathcal{G}_2]_{1,1}\approx\sum_{r=1}^4\left(\sum_{i=0}^2\sum_{p=-5}^{5} d_{r,p}^i\cdot Q_i(p\Delta\omega)\right)^2,
\end{align}
where $ d_{r,p}^i $ is a constant unrelated to $ \Delta\omega $, completing the proof of the part of $ \mathcal{G}_2 $. 

By combining the proofs on $ \mathcal{G}_1 $ and $ \mathcal{G}_2 $, we complete the proof of this lemma.

\section{Proof pf Lemma \ref{lemma:FC}}
\label{app:proof_of_prop3}



We rewrite $ {\rm CRB}_{\rm FC}(\Delta\omega) $ in \eqref{equ:CRB_FC} as a function of $ Q(\Delta\omega) $.
Particularly, through substituting $ \bm{A}=[\bm{a}(\omega_1),\dots,\bm{a}(\omega_L)] $ and $ [\bm{a}(\omega_l)]_n=e^{j\omega_l\varphi_n} $ to $ \bm{\Pi}_{\bm{A}}^{\perp} $ in \eqref{equ:PiA}, we have
\begin{align}
\label{equ:piAQ}
\bm{\Pi}_{\bm{A}}^{\perp}&=\bm{I}-\frac{\bm{A}(\bm{I}-\bm{C}_Q)\bm{A}^H}{N(1-|Q_0|^2)},
\end{align}
where $ \bm{C}_Q $ is denoted by 
\begin{align}
\label{equ:C_Q}
\bm{C}_Q=\begin{bmatrix}
0 & Q_0 \\
Q_0^H & 0
\end{bmatrix}.
\end{align}
Through substituting $ \bm{D} $ in \eqref{equ:D} and $ \bm{\Pi}_{\bm{A}}^{\perp} $ in \eqref{equ:piAQ} to $ \bm{F}={\rm Re}\left\{\bm{D}^H\bm{\Pi}_{\bm{A}}^{\perp}\bm{D}\right\} $ in \eqref{equ:F}, we have 
\begin{align}
\label{equ:FQexp2}
\bm{F}=&{\rm Re}\left\{(\bm{F}_1+\bm{F}_2+\bm{F}_3)\odot\bm{S}\right\},
\end{align}
where 
\begin{align}
\label{equ:F1}
\bm{F}_1&=\sum_{n=1}^N\varphi_n^2\cdot\begin{bmatrix}
1 & Q_2 \\
Q_2^H & 1
\end{bmatrix}, \\
\label{equ:F2}
\bm{F}_2&=-\gamma_0\cdot\begin{bmatrix}
1+|Q_1|^2 & 2Q_1 \\
2Q_1^H & 1+|Q_1|^2
\end{bmatrix}, \\
\label{equ:F3}
\bm{F}_3&=\gamma_0\cdot\begin{bmatrix}
Q_0Q_1^H+Q_0^HQ_1 & Q_0^HQ_1Q_1+Q_0 \\
Q_0Q_1^HQ_1^H+Q_0^H & Q_0Q_1^H+Q_0^HQ_1
\end{bmatrix},
\end{align}
where $ \gamma_0=\frac{(\sum_{n=1}^N\varphi_n)^2}{N(1-|Q_0|^2)} $ and $ \bm{S} $ is defined as
\begin{align}
\label{equ:S}
\bm{S}\equiv
\begin{bmatrix}
|s_1|^2 & s_1^Hs_2 \\
s_2^Hs_1 & |s_2|^2
\end{bmatrix},
\end{align}
completing the proof.

\section{Proof of Lemma \ref{lemma:lemmaBB}}
\label{app:BB}

The following equation is satisfied, which can be directly verified by matrix multiplication:
\begin{align}
\label{equ:matrixinversion_approx}
\widetilde{\bm{B}}\widehat{\bm{B}}=\bm{I}+\bar{\bm{B}},
\end{align}
where
\begin{align}
\label{equ:Btilde}
\widetilde{\bm{B}}&=\begin{bmatrix}
B_2 & & 0 \\
 & \ddots & \\
0 & & B_k
\end{bmatrix}-
{\rm Re}\left\{\begin{bmatrix}
\bm{b}_2^H \\
\vdots \\
\bm{b}_K^H
\end{bmatrix}
\begin{bmatrix}
\bm{b}_2 & \cdots & \bm{b}_K
\end{bmatrix}\right\},
\end{align}
\begin{align}
\label{equ:Bhat}
\widehat{\bm{B}}&=\begin{bmatrix}
\frac{1}{B_2} & & 0 \\
 & \ddots & \\
0 & & \frac{1}{B_k}
\end{bmatrix}+ {\rm Re}\left\{
\begin{bmatrix}
\frac{\bm{b}_2^H}{B_2} \\
\vdots \\
\frac{\bm{b}_K^H}{B_K}
\end{bmatrix}
\begin{bmatrix}
\frac{\bm{b}_2}{B_2} & \cdots & \frac{\bm{b}_K}{B_K}
\end{bmatrix}\right\},
\end{align}
where $ B_2,\dots,B_K>0 $, and the matrix $ [\bar{\bm{B}}]_{k_1-1,k_2-1}=-\sum_{k=2}^K \frac{{\rm Re}\{\bm{b}_{k_1}^H\bm{b}_k\}{\rm Re}\{\bm{b}_k^H\bm{b}_{k_2}\}}{B_kB_{k_2}} $ for $ k_1,k_2=2,\dots,K $.

When $ B_k=[\bm{h}]_{k-1} $ and $ \bm{b}_k=\frac{1}{\sqrt{N}}[\bm{A}^H\bm{H}]_{k-1} $, i.e., $ \widetilde{\bm{B}}=\widetilde{\bm{G}} $, we have
\begin{align}
\label{equ:barBsat}
\underset{K\rightarrow\infty}{\lim}\  [\bar{\bm{B}}]_{k_1-1,k_2-1} = 0,
\end{align}
yielding $ \bar{\bm{B}}\approx\bm{0} $ when $ K $ is large enough.
The reason is shown as follows: Based on the A2 assumption, we have $ |\mathcal{N}_k|/N=1/K $ and $ K $ is large ($ 1/K\approx 0 $).
By substituting \eqref{equ:H} to $ \bar{\bm{B}} $, we have
\begin{align}
\label{equ:Bkk_approx}
[\bar{\bm{B}}]_{k_1-1,k_2-1}&=\sum_{k=2}^K\frac{|\mathcal{N}_{k_1}||\mathcal{N}_{k}|^2|\mathcal{N}_{k_2}|}{N^2|\mathcal{N}_{k}||\mathcal{N}_{k_2}|}\cdot f_{k_1,k_2}(\bm{s},\bm{\omega},Q_0^k) \nonumber \\
&\approx\frac{|\mathcal{N}_{k_1}|}{N}\cdot f_{k_1,k_2}(\bm{s},\bm{\omega},Q_0^k) \nonumber \\
&=\frac{1}{K}\cdot f_{k_1,k_2}(\bm{s},\bm{\omega},Q_0^k),
\end{align}
where $ f_{k_1,k_2}(\bm{s},\bm{\omega},Q_0^k) $ denotes a function of $ \bm{s} $, $ \bm{\omega} $ and $ Q_0^k $ for $ k_1,k_2=2,\dots,K $, and the approximation in \eqref{equ:Bkk_approx} is derived from $ \sum_{k=2}^K|\mathcal{N}_k|=\frac{K-1}{K}N\approx N $.
When $ 1/K\approx 0 $, we have the conclusion in \eqref{equ:barBsat}.

Based on Assumption A2, we have $ \bar{\bm{B}}\approx\bm{0} $, which means
\begin{align}
\label{equ:BB=I}
\widetilde{\bm{B}}\widehat{\bm{B}}\approx\bm{I},
\end{align}
yielding $ \widetilde{\bm{G}}^{-1}\approx\widehat{\bm{B}} $ in \eqref{equ:Bhat} with $ B_k=[\bm{h}]_{k-1} $ and $ \bm{b}_k=\frac{1}{\sqrt{N}}[\bm{A}^H\bm{H}]_{k-1} $, completing the proof. 

\section{Proof of Lemma \ref{lemma:lemma_gradient_up}}
\label{app:up}

The gradient of $ u_p $ w.r.t. $ \mu $ is calculated as 
\begin{align}
\label{equ:up_gradient_cal}
u_p'=\frac{{\rm d}u_p}{{\rm d}\mu}=(-2)^p\cdot\frac{(p+2)(\mu\tilde{s})^{-1}-p\mu\tilde{s}}{\left((\mu\tilde{s})^{-1}+\mu\tilde{s}\right)^{p+2}}.
\end{align}
Since $ u_p $ is either odd or even function of $ \mu $, we assume $ \mu>0 $ without loss of generality.

We prove the 1-st item in Lemma \ref{lemma:lemma_gradient_up}.
Based on \eqref{equ:up_gradient_cal}, \eqref{equ:pro1} is equivalent to  
\begin{align}
\label{equ:twoexp}
2^p\left((p+2)(\mu\tilde{s})^{-1}-p\mu\tilde{s}\right)\leqslant \left((\mu\tilde{s})^{-1}+\mu\tilde{s}\right)^{p+2}.
\end{align}
Apply binomial expansion to the right part of \eqref{equ:twoexp}, yielding 
\begin{align}
\label{equ:rightpart}
\big((\mu\tilde{s})^{-1}+\mu\tilde{s}\big)^{p+2}=\sum_{i=0}^{p+2}B_i(\mu),
\end{align}
where 
\begin{align}
\label{equ:U}
B_i(\mu)=\binom{p+2}{i} 2^{p+2-i}\left((\mu\tilde{s})^{-1}+\mu\tilde{s}-2\right)^i.
\end{align}
Since $ (\mu\tilde{s})^{-1}+\mu\tilde{s}-2\geqslant 0 $, $ B_i(\mu)\geqslant 0 $ for $ i=0,\dots,p+2 $.
Then, we prove that the left part of \eqref{equ:twoexp} is less than $ B_0(\mu)+B_1(\mu) $ when $ p\geqslant3 $, yielding the inequality of \eqref{equ:twoexp}.
Particularly, we have
\begin{align}
\label{equ:minus}
2^p&\left((p+2)(\mu\tilde{s})^{-1}-p\mu\tilde{s}\right)-B_0(\mu)-B_1(\mu) \nonumber \\
&=2^p\left(4(p+1)-(3p+4)\tilde{s}\mu-(p+2)(\tilde{s}\mu)^{-1}\right) \nonumber \\
&\leqslant 2^p\left(4(p+1)-2\sqrt{(3p+4)(p+2)}\right)\leqslant 0,
\end{align}
for $ p\geqslant3, p\in\mathbb{N} $, completing the proof of the 1-st item in Lemma~\ref{lemma:lemma_gradient_up}.

For the 2-nd item in Lemma~\ref{lemma:lemma_gradient_up}, based on \eqref{equ:up_gradient_cal}, we have
\begin{align}
\label{equ:up_pro2}
0\leqslant\underset{|\mu|\rightarrow\infty}{\rm Lim} |u_p'|&\leqslant\underset{|\mu|\rightarrow\infty}{\rm Lim}\frac{2^p(p+2)(\mu\tilde{s})^{-1}}{(\mu\tilde{s})^{p+2}} \nonumber \\
&=\underset{|\mu|\rightarrow\infty}{\rm Lim}\frac{2^p(p+2)}{(\mu\tilde{s})^{p+3}}=0,
\end{align}
completing the proof of the 2-nd item in Lemma~\ref{lemma:lemma_gradient_up}.

\ifCLASSOPTIONcaptionsoff
  \newpage
\fi

\bibliography{reference}
\bibliographystyle{IEEEtran}

\end{document}